\documentclass[a4paper,english,cleveref, autoref, thm-restate,authorcolumns]{lipics-v2019}

\title{Single-use automata and transducers for infinite alphabets} %

\author{Miko\l aj Boja\' nczyk}{Institute of Informatics,\\ University of
Warsaw, Poland}{bojan@mimuw.edu.pl}{}{}  %

\author{Rafa\l{} Stefa\'nski}{Institute of Informatics,\\ University of
Warsaw, Poland}{rafal.stefanski@mimuw.edu.pl}{}{}

\authorrunning{M. Boja\'nczyk and R. Stefa\'nski}%

\Copyright{Miko{\l}aj Boja\'nczyk and Rafa\l\  Stefa\'nski}%



\category{} %

\relatedversion{} %

\supplement{}%

\acknowledgements{Supported by the European Research Council under the European Unions Horizon 2020 research and innovation programme (ERC consolidator grant LIPA, agreement no. 683080).}%

\nolinenumbers %

\hideLIPIcs  %

\EventEditors{John Q. Open and Joan R. Access}
\EventNoEds{2}
\EventLongTitle{42nd Conference on Very Important Topics (CVIT 2016)}
\EventShortTitle{CVIT 2016}
\EventAcronym{CVIT}
\EventYear{2016}
\EventDate{December 24--27, 2016}
\EventLocation{Little Whinging, United Kingdom}
\EventLogo{}
\SeriesVolume{42}
\ArticleNo{23}
\usepackage{url}
\usepackage[all]{xy}
\usepackage{amsmath}
\usepackage{amssymb}
\usepackage{amsthm}
\usepackage{color}
\usepackage{graphicx}
\usepackage{xspace}
\usepackage{enumerate}
\usepackage{alltt}
\usepackage{wrapfig}
\usepackage{scalerel,stackengine}
\usepackage{pifont}
\usepackage{proof}
\usepackage{stmaryrd}
\usepackage{float}
\usepackage{hyperref}
\theoremstyle{plain}

\theoremstyle{plain}

\begin{document}

\maketitle

\begin{abstract}
Our starting point are register  automata for data words, in the style of Kaminski and Francez. We study the effects of the  single-use restriction,  which says that a  register  is emptied immediately after being used. We show that under the single-use restriction,  the theory of automata for data words becomes much more robust. The main results are: (a) five different machine models are equivalent as language acceptors, including one-way and two-way single-use register automata; (b) one can recover some of the algebraic theory of languages over finite alphabets, including a version of the Krohn-Rhodes Theorem; (c) there is also a robust theory of transducers, with four equivalent models, including two-way single use transducers and a variant of streaming string transducers for data words. These results are in contrast with automata for data words without the single-use restriction, where essentially all models are pairwise non-equivalent.
\end{abstract}
\vfill
\pagebreak

\newcommand{\smallatoms}{{\scriptstyle\atoms}}
\newcommand{\onesurt}{\xspace{\footnotesize \textnormal{1}}{\sc det}$\smallatoms$\xspace}
\newcommand{\twosurt}{\xspace{\footnotesize \textnormal{2}}{\sc det}$\smallatoms$\xspace}
\newcommand{\sst}{{\sc sst}\xspace}
\newcommand{\atomlesssst}{{\sc sst}\xspace}
\newcommand{\onedet}{{\xspace{\footnotesize \textnormal{1}}{\sc det}\xspace }}
\newcommand{\twodet}{{\xspace{\footnotesize \textnormal{2}}{\sc det}\xspace }}
\newcommand{\twosst}{\xspace{\footnotesize \textnormal{2}}{\sc sst}\xspace }
\newcommand{\reglist}{{\sc reglist}$\smallatoms$\xspace}
\newcommand{\rwrite}{\overset{\downarrow}{\text{\resizebox{0.02\hsize}{!}{$\square$}}}}
\newcommand{\rread}{\overset{\uparrow}{\text{\resizebox{0.02\hsize}{!}{$\square$}}}}
\newcommand{\hideproof}[1]{#1}

\newcommand{\mso}{{\sc mso}\xspace}
\newcommand{\rgmso}{rigidly guarded {\sc mso}$^\sim$\xspace}
\newcommand{\dmso}{{\sc mso}$^\sim$\xspace}
\newcommand{\fo}{{\sc fo}\xspace}

\newcommand{\smallparagraph}[1]{\smallskip \noindent {\bf #1. }}

\newcommand{\clone}{\mathsf C}
\newcommand{\clonereg}{\mathsf C^{\mathrm{reg}}}
\newcommand{\rank}[1]{\mathrm{rank}(#1)}
\newcommand{\diva}{{\sc{(da)}}}
\newcommand{\homf}{{\sc{(hf)}}}
\newcommand{\fing}{{\sc{(fg)}}}
\newcommand{\eqdef}{\stackrel{\text{\tiny def}}=}
\newcommand{\trees}{{\mathsf{trees}}}
\newcommand{\powerset}{{\mathsf P}}
\newcommand{\muddles}{{\mathsf M}}
\newcommand{\unit}{\mathsf{unit}}
\newcommand{\flatt}{\mathsf{flat}}
\newcommand{\aalg}{\mathbf{A}}
\newcommand{\balg}{\mathbf{B}}
\newcommand{\Nat}{\mathbb N}
\newcommand{\hsp}{{\sc hsp \xspace}}

\newcommand{\hs}{{\sc hs \xspace}}
\newcommand{\aut}{\mathcal A}
\newcommand{\nodes}{\mathsf{nodes}}
\newcommand{\set}[1]{\{#1\}}
\newcommand{\dom}{\mathsf{dom}}
\newcommand{\game}{\mathsf G}
\newcommand{\mult}{\mathsf{mult}}
\newcommand{\parfun}{\rightharpoonup}
\newcommand{\monad}{\mathsf M}
\newcommand{\facto}{\mathsf F}
\newcommand{\poly}[2]{\mathsf{pol}_{#1}#2}
\newcommand{\Aa}{\mathcal A}
\newcommand{\map}{\mathsf{map}}
\newcommand{\sub}{\mathsf{sub}}
\newcommand{\SplitLemma}{\hyperref[lem:split-lemma]{Split Lemma}\xspace}

\newcounter{ourexamplecounter}
\newenvironment{myexample}{
\medskip

\refstepcounter{ourexamplecounter}
\smallskip\noindent{\textbf{{Example \arabic{ourexamplecounter}. }}}}{
$\Box$ \smallskip 
}

\newcounter{runningcounter}
\newenvironment{running}{
\medskip

\refstepcounter{runningcounter}
\smallskip\noindent{\textbf{{Running Example \arabic{runningcounter}. }}}}{
$\Box$ \smallskip 
}

\newcommand{\atoms}{\mathbb A}

\newcommand{\Jj}{\mathcal J}
\newcommand{\Rr}{\mathcal R}
\newcommand{\Ll}{\mathcal L}
\newcommand{\Hh}{\mathcal H}

\newcommand{\mypic}[1]{
	\begin{center}
		\includegraphics[page=#1,scale=0.4]{pics}
	\end{center}
}

\newcommand{\picfromfile}[2] {
  \begin{center}
    \includegraphics[width=#2\columnwidth]{#1}
  \end{center}
}

\newcommand{\finsupfun}{\stackrel {\mathrm{fs}}\to }
\newcommand{\red}[1]{{\color{red}#1}}
\newcommand{\blue}[1]{{\color{blue}#1}}

\stackMath
\newcommand\reallywidehat[1]{%
\savestack{\tmpbox}{\stretchto{%
  \scaleto{%
    \scalerel*[\widthof{\ensuremath{#1}}]{\kern-.6pt\bigwedge\kern-.6pt}%
    {\rule[-\textheight/2]{1ex}{\textheight}}%
  }{\textheight}%
}{0.5ex}}%
\stackon[1pt]{#1}{\tmpbox}%
}
\parskip 1ex

\newcommand{\cmark}{\text{\ding{51}}}
\newcommand{\xmark}{\text{\ding{55}}}

\newcommand{\applyendmarkers}[1]{\vdash\!\! #1 \!\! \dashv}

\newcommand{\rep}[1]{\textrm{rep}(#1)}
\newcommand{\supp}[1]{\textrm{supp}(#1)}
\newcommand{\run}{\mathrm{run}}

\newcommand{\threeinclusion}[3]
{\small
    \text{#1} \circ \text{#2} \ \subseteq \ \text{#3}
}

\newcommand{\twoinclusion}[2]
{\small
    \text{#1}  \ \subseteq \ \text{#2}
}

\newcommand{\twoequality}[2]
{\small
    \text{#1}  \ =   \ \text{#2}
}

\newcommand{\boolset}{\set{\text{``yes'',``no''}}}
\newcommand{\alphabet}{\mathbb B}
\newcommand{\deatomise}[1]{\underline{#1}}

\section{Introduction}
\label{sec:introduction}

One of the appealing features of regular languages for finite alphabets is the robustness of the notion: it can be characterised by many equivalent models of automata (one-way, two-way, deterministic, nondeterministic, alternating, etc.), regular expressions, finite semigroups, or monadic second-order logic. A similar robustness appears for transducers, see~\cite{filiotR16} for a survey; particularly for the class of {regular string-to-string functions}, which can be characterised using deterministic two-way transducers, streaming string transducers, or \mso transductions.

This robustness vanishes for infinite alphabets. We consider infinite alphabets that are constructed using an infinite set $\atoms$ of atoms, also called data values. Atoms can only be compared for equality. The literature for  infinite alphabets is full of depressing diagrams like~\cite[Figure~1]{nevenFiniteStateMachines2004} or~\cite[p.~24]{bojanczyk_slightly2018}, which describe countless models that satisfy only trivial relationships such as deterministic $\subseteq$ nondeterministic, one-way $\subseteq$ two-way, etc. 

This lack of robustness has caused several authors to ask if there is a notion of ``regular language'' for infinite alphabets; see~\cite[p. 703]{bjorklundNotionsRegularityData2010} or~\cite[p.~2]{bojanczykAutomataDataWords2010}. This question was probably rhetorical, with the assumed answer being ``no''.
In this paper, we postulate a ``yes'' answer.  The main theme is register automata, as introduced by Kaminski and Francez~\cite{kaminskiFiniteMemoryAutomata1994}, but with the single-use restriction, which says that immediately after a register is used, its value is destroyed. As we show in this paper, many automata constructions, which fail for unrestricted register automata, start to work again in the presence of the single-use restriction. 

 Before describing the  results in the paper, we illustrate the single-use restriction.

\begin{myexample}\label{ex:three-letters}
     Consider the language ``there are at most three distinct letters in the input word, not counting repetitions'', over  alphabet $\atoms$. There is a natural  register automaton which recognises this language: use three registers to store the distinct atoms that have been seen so far, and if a fourth atom comes up, then reject.  
     This automaton, however, violates the single-use restriction, because each new input letter is compared to all the registers.
 
\begin{wrapfigure}{r}{0.4\textwidth}
        \mypic{13}
    \end{wrapfigure}
 Here is a solution that respects the single-use restriction. The idea is that once the automaton has seen three distinct letters $a,b,c$, it stores them in six registers as explained in the picture on the right.  Assume that a new input letter $d$ is read. The behaviour of the automaton (when it already has three atoms in its registers) is explained in the  flowchart in Figure~\ref{fig:flowchart}.
    
 A similar flowchart is used for the corner cases when the automaton has seen less than three letters so far.    
    \begin{figure}
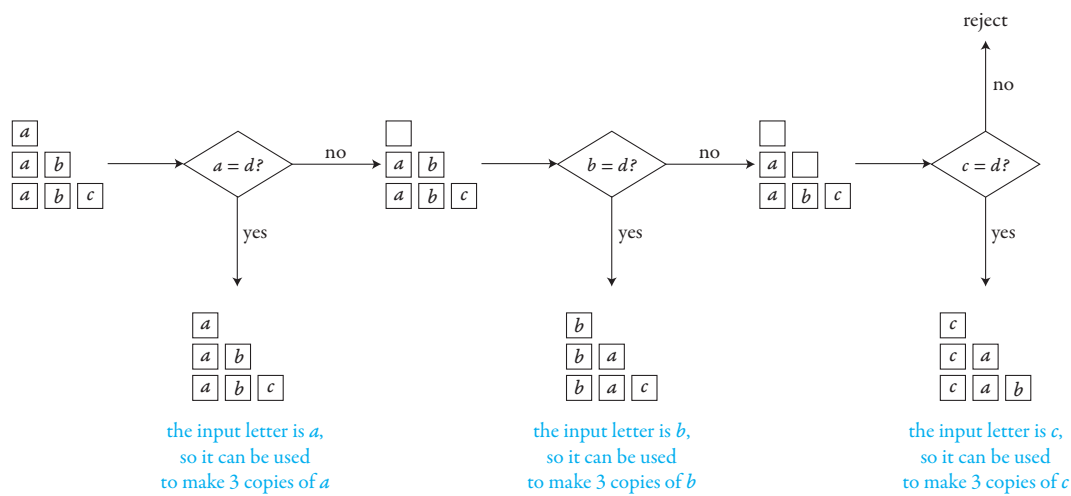

        \mypic{14}
        \caption{Updating the six  registers.}
        \label{fig:flowchart}
    \end{figure}
\end{myexample}    

Our first main result, Theorem~\ref{thm:one-way-two-way} (in Section
\ref{sec:languages}), says  that the following models recognise the same languages over infinite alphabets:
\begin{enumerate}
    \item deterministic one-way single-use automata;
    \item deterministic two-way single-use automata;
    \item  orbit-finite monoids \cite{bojanczykNominalMonoids2013}; 
    \item  \rgmso \cite{DBLP:journals/corr/ColcombetLP14}; 
    \item string-to-boolean regular list functions with atoms.
\end{enumerate}
The equivalence of the models in items 3 and 4 was shown in~\cite{DBLP:journals/corr/ColcombetLP14}; the remaining models and their equivalences are new (item 5 is an extension of the regular list functions from~\cite{bojanczykRegularFirstOrderList2018}).

Just like their classical versions, one-way and two-way single-use automata are equivalent as language acceptors, but they are no longer equivalent as transducers. For example, a two-way single-use transducer can reverse the input string, which is impossible for a one-way single-use transducer.  In Sections \ref{sec:su-sequential}
and \ref{sec:two-way-transducers} we develop the theory of single-use transducers:

In Section~\ref{sec:su-sequential}, we  investigate single-use one-way transducers. For finite alphabets, one of the most important results about one-way transducers is the Krohn-Rhodes Theorem~\cite{Krohn1965}, which says that every Mealy machine (which is a length preserving one-way transducer) can be decomposed into certain ``prime'' Mealy machines.  We show that the same can be done for infinite alphabets, using  a single-use extension of Mealy machines. The underlying prime machines are the machines from the original Krohn-Rhodes theorem, plus one additional register machine which moves atoms to later positions.

In Section~\ref{sec:two-way-transducers}, we investigate single-use two-way transducers, and show that the corresponding class of string-to-string functions enjoys similar robustness properties as the languages discussed in Theorem~\ref{thm:one-way-two-way}, with four models being equivalent:
\begin{enumerate}
    \item single-use two-way transducers;
    \item an atom extension of streaming string transducers~\cite{alurExpressivenessStreamingString2010};
    \item string-to-string regular list functions with atoms;
    \item compositions of certain ``prime two-way machines''
    (Krohn \& Rhodes style).
\end{enumerate}
We also show other good properties of the string-to-string functions in the above items, including closure under composition (which follows from item 4) and decidable equivalence. 

Summing up, the single-use restriction allows us to identify    languages and string-to-string functions with infinite alphabets, which share the robustness and good mathematical theory  usually associated with regularity for finite alphabets.

Due to space constraints, and a large number of results, virtually all of the proofs are in an appendix. We use the available space to explain and justify the many new models that are introduced.
\vfill 
\section{Automata and transducers with atoms}
\label{sec:the-model}
For the rest of the paper, fix an infinite set $\atoms$, whose elements are called \emph{atoms}.  Atoms will be used to construct infinite alphabets. Intuitively speaking, atoms can only be compared for equality. It would be interesting enough to consider alphabets of the form $\atoms \times \Sigma$, for some finite
$\Sigma$, as is typically done in the literature on data words~\cite[p.~1]{bojanczykAutomataDataWords2010}. However, in the proofs, we use more complicated sets, such as the set $\atoms^2$ of pairs of atoms, the set $\atoms + \set{\vdash,\dashv}$ obtained by adding two endmarkers to the atoms, or the co-product (i.e.~disjoint union)   $\atoms^2 + \atoms^3$. This motivates the following definition.

\begin{definition}
    A \emph{polynomial orbit-finite set\footnote{The name ``orbit-finite'' is used because the above definition is a special case of orbit-finite sets discussed later in the paper, and the name ``polynomial'' is used  to underline that the sets are closed under products and co-products.
    }} is any set that can be obtained from $\atoms$ and singleton sets by means of finite products and co-products (i.e.~disjoint unions).
\end{definition}

We only care about properties of such sets that are stable under atom automorphisms, as described below. 
Define an \emph{atom automorphism} to be any bijection $\atoms \to \atoms$. (This notion of automorphism formalises the intuition that atoms can only be compared for equality). Atom automorphisms form a group. There is a natural action of this group on polynomial orbit-finite sets: for elements of  $\atoms$ we  apply the atom automorphism, for singleton sets the action is trivial, and for other polynomial orbit-finite sets the action is lifted inductively along $+$ and $\times$ in the natural way.
  Let $\Sigma$ and $\Gamma$ be sets equipped with an action of the group of atom automorphisms -- in particular, these could be  polynomial orbit-finite sets. A  function $f : \Sigma \to \Gamma$ is called \emph{equivariant} if $f(\pi(x)) = \pi(f(x))$ holds for every $x \in \Sigma$ and every atom automorphism $\pi$. \label{page:equivariant} The general idea is that equivariant functions can only talk about equality of atoms. In the case of polynomial orbit-finite sets, equivariant functions can also be finitely represented using quantifier-free formulas~\cite[Lemma 1.3]{bojanczyk_slightly2018}. 

\smallparagraph{The model} We now describe the single-use machine models discussed in this paper. There are four variants: machines can be one-way or two-way, and they can recognise languages or compute string-to-string functions.  We begin with  the most general form -- two-way string-to-string functions --  and define  the other models as special cases.

 The machine reads the input string, extended with left and right endmarkers $\vdash,\dashv$. It uses registers to store atoms that appear in the input string. A register can store either an atom, or the undefined value $\bot$.  The single-use restriction, which is  written in red below, says that a register is set to $\bot$ immediately after being used. 

\newcommand{\sigmaendmarker}{\Sigma +\set{\vdash,\dashv}}
\begin{definition}\label{def:the-transducer-model}
    The syntax of a {\em two-way single-use transducer\footnote{Unless otherwise noted, all transducers and automata considered in this paper are deterministic. The theory of nondeterministic single-use models seems to be less appealing. }} consists of
    \begin{itemize}
        \item input and output alphabets $\Sigma$ and $\Gamma$, both  polynomial orbit-finite sets;
        \item a finite set of states $Q$, with a distinguished initial state $q_0 \in Q$;
        \item a finite set $R$ of register names;
        \item a \emph{transition function} which maps each state $q \in Q$ to an element of:
        \begin{align*}
             \underbrace{\text{questions}}_{\text{question that is asked}} \times \underbrace{(Q \times \text{actions})}_{\substack{\text{what to do if the}\\ \text{question has a yes answer}}} \times \underbrace{(Q \times \text{actions})}_{\substack{\text{what to do if the}\\ \text{question has a no answer}}}
        \end{align*}
        where the allowed questions and actions  are taken from the following toolkit:
        \begin{enumerate}
            \item \emph{Questions.}
            \begin{enumerate}
                \item Apply  an equivariant function $f : \sigmaendmarker \to \set{\text{yes, no}}$ to the letter under the head, and return the answer.
                \item Are the atoms stored in registers $r_1, r_2$  equal and defined?  If any of these registers is undefined, then the run immediately stops and rejects\footnote{By remembering in the state which registers are defined, one can modify an automaton so that this never happens.}.  \red{This question has the side
                effect of setting the values of $r_1$ and $r_2$ to $\bot$}.
            \end{enumerate}
            \item \emph{Actions.}
        \begin{enumerate}  
            \item Apply an   equivariant function $f : \sigmaendmarker \to \atoms + \bot$ to the letter under the head, and store the result in  register $r \in R$.
            \item \label{action:output} Apply an  equivariant function $f : \atoms^k \to \Gamma$ to the contents of distinct registers $r_1,\ldots,r_k \in R$, and append the result to the output string. If any of the registers is undefined, stop and reject.   \red{This
                  action has the side effect of setting the values of $r_1, r_2, \ldots, r_k$ to~$\bot$.}
                  \item  \label{action:move-head} Move the head to the previous/next input position.
            \item \label{acction:accept} Accept/reject and finish the run.
        \end{enumerate}
        \end{enumerate}
    \end{itemize}
\end{definition}
The semantics of the transducer is a partial function from strings over the input alphabet to strings over the output alphabet. 
Consider a string of the form  $\applyendmarkers w$ where $w \in \Sigma^*$. A \emph{configuration} over such a string consists of (a) a position in the  string; (b)
   a state; (c) a register valuation, which is a function of type $R \to \atoms + \bot$; (d) an output string, which is a string over the output alphabet. 
A \emph{run} of the transducer is defined to be a sequence of configurations, where consecutive configurations are related by applying the transition function in the natural way. The \emph{output} of a run is defined to be the contents of the output string in the last configuration.   An \emph{accepting configuration} is one which executes the accept action from item~\ref{acction:accept} -- accepting configurations have no successors. 
The \emph{initial configuration} is a configuration where the head is over the left endmarker $\vdash$, the state is the initial state, the register valuation maps all registers to the undefined value, and the output string is empty.  An \emph{accepting run} is a run that begins in the initial configuration and ends in an accepting one.  By determinism, there is at most one accepting run. The semantics of the transducer is defined to be the partial function $\Sigma^* \to \Gamma^*$, which inputs $w \in \Sigma^*$ and returns the output of   the accepting run over $\applyendmarkers  w$. If there is no accepting run, $f(w)$ has no value.

\smallparagraph{Special cases} A \emph{one-way single-use transducer} is the special case of Definition~\ref{def:the-transducer-model} which does not use the ``previous'' action from item~\ref{action:move-head}.
A \emph{two-way single-use automaton} is the special case  which does not use the output actions from item~\ref{action:output}. The \emph{language} recognised by such an automaton is defined to be the set of words which admit an accepting run. A \emph{one-way single-use automaton} is the special case of a two-way single-use automaton, which does not use the ``previous'' action from item~\ref{action:move-head}.
\section{Languages recognised by single-use automata}
\label{sec:languages}
In this section we discuss languages recognised by single-use automata. The main result is that  one-way and two-way single-use automata recognise the same languages, and furthermore these are the same languages that are recognised by orbit-finite monoids~\cite{bojanczykNominalMonoids2013}, the logic \rgmso~\cite{DBLP:journals/corr/ColcombetLP14}, and a new model called regular list functions with atoms, that will be defined in Section~\ref{sec:two-way-transducers}.

\smallparagraph{Orbit-finite monoids}  We begin by defining orbit-finite sets and orbit-finite monoids, which play an important technical role in this paper. For more on orbit-finite sets, see the lecture notes~\cite{bojanczyk_slightly2018}.
For a tuple $\bar a \in \atoms^*$, an $\bar a$-automorphism is defined to be any atom automorphism that maps $\bar a$ to itself. 
Consider set $X$ equipped with an action of the group of atom automorphisms. We say that $x \in X$ is \emph{supported} by a tuple of atoms $\bar a \in \atoms^*$ if $\pi(x)=x$ holds for every $\bar a$-automorphism $\pi$. We say that a subset of  $X$ is $\bar a$-supported if it is an $\bar a$-supported element of the powerset of $X$; similarly we define supports of relations and functions. We say that $x$ is finitely supported if it is supported by some tuple $\bar a \in \atoms^*$. Define the  \emph{$\bar a$-orbit of $x$} to be its orbit under the action of the group of $\bar a$-automorphisms. 

\begin{definition}[Orbit-finite sets] \label{def:orbit-finite}
    Let $X$ be a set equipped with an action of atom automorphisms. A subset  $Y \subseteq X$ is called \emph{orbit-finite} if (a) every element of $Y$ is finitely supported; and (b) there exists some $\bar a \in \atoms^*$ such that $Y$ is a union of finitely many $\bar a$-orbits.
\end{definition}

An equivariant orbit-finite set is the special case where the tuple $\bar a$ in item (b) is empty. The polynomial orbit-finite sets from Section~\ref{sec:the-model} are a special case of equivariant orbit-finite sets\footnote{The converse does not hold -- there exist sets that are equivariant orbit finite but not polynomial orbit finite e. g. the set of unordered pairs of atoms:
$\set{\set{a,b} \;\;|\;\; a, b\in \atoms, \; a \neq b}$.}.
The following notion was introduced in~\cite[Section 3]{bojanczykNominalMonoids2013}.
\begin{definition}[Orbit-finite monoid]
    An \emph{orbit-finite monoid} is a monoid where the underlying set is  orbit-finite, and the monoid operation is finitely supported. Let $\Sigma$ be an orbit-finite set. We say that a language $L \subseteq \Sigma^*$ is \emph{recognised} by an orbit-finite monoid $M$ if there is a   finitely supported  monoid
    morphism $h : \Sigma^* \to M$ and a finitely supported accepting set $F \subseteq M$ such that $L$ contains exactly the words whose image under $h$ belongs to $F$. 
\end{definition}

In this paper, we are mainly interested in the case where both the morphism and the accepting set are equivariant. In this case, it follows that the alphabet $\Sigma$, the image of the morphism, and the recognised language all also have to be equivariant.

The structural theory of orbit-finite monoids was first developed in~\cite{bojanczykNominalMonoids2013}, where it was shown how the classical results about Green's relations for finite monoids extend to the orbit-finite setting. This theory was further investigated in~\cite{DBLP:journals/corr/ColcombetLP14}, including a lemma stating  that every orbit-finite group is necessarily finite. In the appendix of this paper we build on these results, to prove an orbit-finite version of the Factorisation Forest Theorem of Simon~\cite[Theorem 6.1]{simonFactorizationForestsFinite1990}, which is used in proofs of Theorems~\ref{thm:one-way-two-way} and~\ref{thm:mealy-machine-krohn-rhodes}.

\smallparagraph{Main theorem about languages} We are now ready to state Theorem~\ref{thm:one-way-two-way}, which is our main result about languages.

\begin{theorem}
    \label{thm:one-way-two-way} Let $\Sigma$ be a polynomial orbit-finite set. The following conditions are equivalent for every language $L \subseteq \Sigma^*$:
    \begin{enumerate}
        \item \label{languages:one-way} $L$ is recognised by a single-use one-way automaton;
        \item \label{languages:two-way} $L$ is recognised by a single-use two-way automaton;
        \item \label{languages:monoids} $L$ is  recognised by an orbit-finite monoid, with an equivariant morphism and an equivariant accepting set;
        \item \label{languages:rgmso} $L$ can be defined in the \rgmso logic;
        \item  \label{languages:list-function}
        $L$'s characteristic function $\Sigma^* \to \set{yes,no}$ is an orbit-finite regular list function.
    \end{enumerate}
\end{theorem}
The equivalence of items~\ref{languages:rgmso} and~\ref{languages:monoids} has been proved in~\cite[Theorems 4.2 and 5.1]{DBLP:journals/corr/ColcombetLP14}, and since we do not use \rgmso outside of the this theorem, we do not give a definition here
(see~\cite[Section 3]{DBLP:journals/corr/ColcombetLP14}). The orbit-finite regular list functions from item~\ref{languages:list-function} will be defined in Section~\ref{sec:two-way-transducers}.
The proof outline for Theorem \ref{thm:one-way-two-way} is given in the following diagram
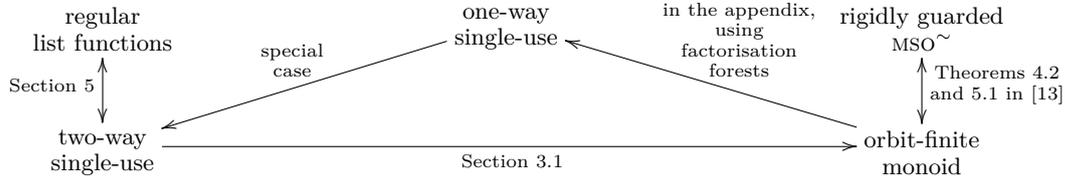
\newcommand{\txtlabel}[1]{\scriptsize{\txt{#1}}}
\begin{center}
    \small
    \xymatrix@C=3.5cm{
        \txt{regular \\list functions} \ar@{<->}[d]_{\txtlabel{Section~\ref{sec:two-way-transducers}}}
         & \txt{one-way\\ single-use} \ar[dl]_{\txtlabel{special\\ case}}  &
         \txt{rigidly guarded \\ \mso$\!\!^\sim$} \ar@{<->}[d]^{\txtlabel{Theorems 4.2 \\ and 5.1 in~\cite{DBLP:journals/corr/ColcombetLP14}}} \\
     \txt{two-way\\ single-use} \ar[rr]_{\txtlabel{Section~\ref{sec:automata-to-semigroups}}} & & 
     \txt{orbit-finite\\ monoid}  \ar[ul]_{\txtlabel{in the appendix,\\ using \\ factorisation \\forests}}
    }
\end{center}
All equivalences in the theorem are effective, i.e.~there are algorithms implementing the  conversions between any of the models.

The single-use restriction is crucial in the  theorem. Automata without the single-use restriction -- call them multiple-use -- only satisfy the trivial inclusions:
\begin{align*}
\text{single-use} \!\!\!\! \stackrel{\txt{\scriptsize first letter appears again\\\scriptsize \cite[Exercise 91]{bojanczyk_slightly2018}}}\subsetneq  \!\!\!\! \text{one-way multiple-use} \!\!\!\!  \stackrel{\txt{\scriptsize some letter appears twice\\ \scriptsize \cite[Example 11]{kaminskiFiniteMemoryAutomata1994}}}\subsetneq \!\!\!\!  \text{two-way multiple-use}.
\end{align*}
Two-way multiple-use automata have an undecidable emptiness problem~\cite[Theorem 5.3]{nevenFiniteStateMachines2004}. For one-way (even multiple-use) automata,  emptiness  is decidable and even tractable  in a suitable parametrised understanding~\cite[Corollary 9.12]{bojanczyk_slightly2018}. We leave open the following question: given a one-way  multiple-use automaton, can one decide if there is an equivalent automaton that is single-use (by Theorem~\ref{thm:one-way-two-way}, it does not matter whether one-way or two-way)?

\subsection{From two-way automata to orbit-finite monoids}
 \label{sec:automata-to-semigroups} 
 In this section, we show   the implication \ref{languages:two-way} $\Rightarrow$ \ref{languages:monoids} of Theorem~\ref{thm:one-way-two-way}. (This  is the only proof of the paper that is not relegated to the appendix -- we chose it, because it illustrates the importance of the single-use restriction).
   The implication states that the language of every single-use two-way automaton can also be recognised by an equivariant homomorphism into an orbit-finite monoid. In the proof, we use the  Shepherdson construction
   for two-way
   automata~\cite{shepherdson1959reduction} and show that, thanks to the single-use restriction, it produces monoids which are orbit-finite.

Consider a two-way single-use automaton, with $k$ registers and let $Q$ be the set of its states.
For a string over the input alphabet (extended with endmarkers), define its \emph{Shepherdson profile} to be the function of the type
\begin{align*}
     \overbrace{Q \times (\atoms+\bot)^k}^{\substack{\text{state and register}\\ \text{valuation at the}\\ \text{start of the run}}}
     \times \overbrace{\set{\leftarrow, \rightarrow}}^{\substack{\text{does the run}\\ \text{enter from the}\\ \text{left or right}}}  
     \qquad \to  \qquad \set{\text{accept, loop}} +  (\overbrace{Q \times (\atoms+\bot)^k}^{\substack{\text{state and register}\\ \text{valuation at the}\\ \text{end of the run}}} \times \overbrace{\set{\leftarrow, \rightarrow}}^{\substack{\text{does the run}\\ \text{exit from the}\\ \text{left or right}}})
\end{align*}
 that describes runs of the automaton in the natural way (see~\cite[Proof of Theorem 2]{shepherdson1959reduction}).  The run is taken until the automaton either exits the string from either side, accepts, or enters an infinite loop. 
  By the same reasoning as in Shepherdson's proof, one can equip the set of Shepherdson profiles with a monoid structure   so that the function which maps a word to its Shepherdson profile becomes  a monoid homomorphism. We use the name Shepherdson monoid for the resulting monoid (it only contains the `achievable' profiles -- the image of $\Sigma^*$).  It is easy to see that whether a word is accepted  depends only on an equivariant property of its Shepherdson profile, and therefore the language recognised by the automaton is also recognised by the Shepherdson monoid.

  It remains to show that the  Shepherdson monoid is orbit-finite, which is the main part of the proof. Unlike the arguments so far, this part of the proof relies on the single-use restriction.  To illustrate  this, we give an example of a  one-way automaton that is not single-use and whose Shepherdson monoid  is not orbit-finite. 
  
  \begin{myexample} \label{ex:appears-again} Consider the language over $\atoms$ of words whose first letter appears again. This language is not recognised by any orbit-finite monoid~\cite[Exercise 91]{bojanczyk_slightly2018}, but it is  recognised by a multiple-use one-way automaton, which stores the first letter in a register, and then compares this register with all remaining letters of the input word. For this automaton, the Shepherdson profile  needs to remember all of the distinct letters that appear in the word. In particular, if two words have different numbers of distinct letters, then  their Shepherdson profiles  cannot be in the same orbit. Since input strings can contain arbitrarily many distinct letters, the Shepherdson monoid
  of this automaton is not orbit-finite. 
  \end{myexample}

\begin{lemma}\label{lem:small-support}
    For every single-use two-way automaton
    there is some $N \in \Nat$ such that every Shepherdson profile  is supported by at most $N$ atoms. 
\end{lemma}

Before proving the lemma, we use it to show that the Shepherdson monoid is orbit-finite. In Section \ref{sec:functions-with-limited-support} of the appendix, we show that if an equivariant set consists of functions from one orbit-finite set to another orbit-finite set (as is the case for the underling set in the Shepherdson monoid)  and all functions in the set have supports of bounded size (as is the case thanks to Lemma \ref{lem:small-support}), then the set is orbit-finite. This leaves us with proving Lemma~\ref{lem:small-support}.

\begin{proof} Define a \emph{transition} in a run to be a pair of consecutive configurations. Each transition has a corresponding question and action. 
      A transition  in  a run  is called \emph{important} if its question or action  involves a  register that has  not appeared in any action or question of the run. The number of important transitions is bounded by $k$ -- the number of registers. The crucial observation, which relies on the single-use restriction, is that  if the input word, head position, and state are fixed
    (but not the register valuation), 
      then the sequence of actions in  the corresponding run depends only on the answers to the questions in the important transitions. This is described in more detail below.

      Fix a choice of the following parameters: (a) a string over the input alphabet that might contain endmarkers; (b) an entry point of the automaton -- either the left or the right end of the word; (c) a state of the automaton. We do not fix the register valuation. 
For a register valuation $\eta$, define $\rho(\eta)$ to be the run which begins in the configuration described by the parameters (abc) together with  $\eta$,  and which is maximal, i.e.~it ends when the automaton either accepts, rejects, or tries to leave the fixed string. For $i \in \set{0,1,\ldots,k}$ define $\alpha_i(\eta)$ to be the sequence of actions that are performed in the maximal prefix of the run $\rho(\eta)$ which uses at most $i$ important transitions. 
The  crucial observation that was stated at the beginning of this proof  is that once the parameters (abc) are fixed, then the sequence of actions $\alpha_i(\eta)$ depends only on the answers to the questions asked in   the first $i$ important transitions. In particular, the function $\alpha_i$ has at most $2^i$ possible values.  Furthermore, by a simple induction on $i$, one can show the following claim.

\begin{claim}
    The function  $\alpha_i$ is supported by at most  $2^{i+1}$ atoms.
\end{claim}

Since there are at most $k$ important transitions in a run, the above claim implies  that, for every fixed choice of parameters  (abc), at most $2^{k+1}$ atoms are needed to support the function which maps $\eta$ to the sequence of actions in the run $\rho(\eta)$. In the arguments for the  Shepherdson profile for a  fixed word $w$, parameter (b) can have  two values (first or last position) and parameter (c) can have at most $|Q|$ values.  Therefore, at most $2|Q|2^{k+1}$ atoms are needed to support the function which takes an argument as in the Shepherdson profile, and returns the sequence of actions in the corresponding run. The lemma follows.  
\end{proof}

\section{A Krohn-Rhodes decomposition of one-way transducers with atoms}
\label{sec:su-sequential}
In this section, we present a decomposition result for single-use one-way transducers, which is a version of the celebrated Krohn-Rhodes Theorem~\cite[p.~454]{Krohn1965}. We think that this result gives further evidence for the good structure of single-use models. In the next section, we  give a similar decomposition result for two-way single-use transducers which will be used to prove the equivalence of several other characterisations of the two-way model. 

We begin by describing the classical Krohn-Rhodes Theorem. 
A \emph{Mealy machine} is a deterministic one-way length-preserving transducer, which is obtained from a deterministic finite automaton by labelling transitions with output letters and ignoring accepting states.
The  Krohn-Rhodes Theorem says that every function computed by a  Mealy machine is a composition of functions computed by certain prime  Mealy machines (which are called \emph{reversible} and \emph{reset} in~\cite[Chapter 6]{abrahamAlgebraicTheoryAutomata1968}). In this section, we prove a version of this theorem for orbit-finite alphabets; this version relies crucially on the single-use restriction. To distinguish the original model of Mealy machines from the single-use model described below, we will use the name \emph{classical Mealy machine} for the Mealy machines in the original Krohn-Rhodes Theorem, i.e.~the alphabets and state spaces are finite.

Define a \emph{single-use Mealy machine}  to have the same syntax as in Definition~\ref{def:the-transducer-model}, with the following differences: there are no ``next/previous'' actions from item~\ref{action:move-head}, but the output action from item~\ref{action:output} has the side effect of moving the head to the next position. A consequence is that a Mealy machine is length-preserving, i.e.~it outputs exactly one letter for each input position. Furthermore, there are no endmarkers and no ``accept'' or ``reject'' actions from item~\ref{acction:accept}; the automaton begins in the first input position and accepts immediately once its head leaves the input word from the right.

\newcommand{\propignore}{\epsilon}
\newcommand{\propload}{\downarrow}
\begin{myexample} \label{ex:atom-propagation} Define  \emph{atom propagation} to be the following length-preserving function. The input alphabet is $\atoms + \set{\propignore, \propload}$ and the output alphabet is  $\atoms + \bot$.   If a  position $i$ in the input string has label $\propload$ and there is some (necessarily unique) position $j < i$ with an atom label such that all positions strictly between $j$ and $i$ have label $\propignore$, then the output label of position $i$ is the atom in input position $j$. For all other input positions, the output label is $\bot$.  
Here is an example of atom propagation:
    \begin{align*}
    \begin{array}{rccccccccccccccccccc}
        \text{input} & 1&2&\propignore & \propignore & \propload & \propload & 3 & \propignore & \propignore & \propload & \propignore & \propload \\
        \text{output} & \bot & \bot & \bot & \bot & 2 &  \bot & \bot & \bot & \bot & 3 & \bot & \bot
    \end{array}
    \end{align*}
    Atom propagation is computed by a single-use Mealy machine, which stores the most recently seen atom in a register, and outputs the register at the nearest appearance of   $\propload$.
\end{myexample}

The following example illustrates some of the technical
difficulties with single-use Mealy machines:
It is often useful to consider a Mealy machine that
computes the run of another Mealy machine --
it decorates every input position with the state and
the register valuation that the Mealy machine
will have after reading the input up to
(but not including) that position. The following example shows that the single-use restriction makes
this construction impossible.
\begin{example}
        \label{ex:no-compute-state}
        Consider the single-use Mealy machine  that implements the atom propagation function from Example~\ref{ex:atom-propagation}.  This machine  has only one register. Every time it sees an atom value, it stores the value in the register and every time it  sees $\downarrow$, and the register is non-empty,
        the machine outputs the register's content.
        We claim that the run of this machine cannot be computed by a Mealy machine. If it could, we would be able to use it to construct a Mealy machine that
        given a word over $\atoms$,
        equips every position with the atom from the first position. This would easily lead to a construction of a single-use automaton
        for the language ``the first letter appears again'' (from Example \ref{ex:appears-again})
        which, as we already know, is impossible.  
\end{example}

The Krohn-Rhodes Theorem, both in its original version and in our orbit-finite version below, says that every Mealy machine can be decomposed into prime functions  using two types of composition:
\begin{align*}
    \infer[\text{sequential}]
{ \Sigma^* \stackrel {g \circ f }\longrightarrow \Delta^*}
{ \Sigma^* \stackrel f \longrightarrow \Gamma^* &  \Gamma^* \stackrel g \longrightarrow \Delta^*}
\qquad 
\infer[\text{parallel}]
{ (\Sigma_1 \times \Sigma_2)^* \stackrel {f_1|f_2}\longrightarrow (\Gamma_1 \times \Gamma_2)^*}
{ \Sigma_1^* \stackrel {f_1} \longrightarrow \Gamma_1^* &  \Sigma_2^* \stackrel {f_2} \longrightarrow \Gamma_2^*}
\end{align*}
The sequential composition is simply function composition. The parallel composition -- which only makes sense for length preserving functions  -- applies the function $f_i$ to the $i$-th projection of the input string. 
\begin{theorem}\label{thm:mealy-machine-krohn-rhodes}
    Every total function computed by a single-use Mealy machine can be obtained, using sequential and parallel composition, from the following \emph{prime functions}:
    \begin{enumerate}
        \item \emph{Length-preserving homomorphisms.} Any function of type $\Sigma^* \to \Gamma^*$, where $\Sigma$ and $\Gamma$ are polynomial orbit-finite, obtained by lifting to strings an equivariant function of type $\Sigma \to \Gamma$.
        \item \emph{Classic Mealy machines.} Any function computed by a classical Mealy machine.
        \item \emph{Atom propagation.} The atom propagation function from Example~\ref{ex:atom-propagation}.
    \end{enumerate}
\end{theorem}
By the original Krohn-Rhodes theorem, classical Mealy machines can be further decomposed.

Define a \emph{composition of primes} to be any function that can be obtained from the prime functions by using a parallel and sequential composition. In this terminology, Theorem~\ref{thm:mealy-machine-krohn-rhodes} says that every function computed by a single-use Mealy machine is a composition of primes. The converse is also true: every prime function is  computed by a single-use Mealy machine, and single-use Mealy machines are closed under both kinds of composition (for details see Section~\ref{ap:mealy-primes} of the appendix).

\smallparagraph{Decomposition of single-use one-way transducers} A Mealy machine is the special case of a single-use one-way transducer which is  length preserving, and does not see an endmarker. A corollary of Theorem~\ref{thm:mealy-machine-krohn-rhodes} is that, in order to generate all total functions computed by single-use one-way transducers, it is enough to add two items to the list of prime functions from Theorem~\ref{thm:mealy-machine-krohn-rhodes}: (a) a function $w \mapsto w\!\! \dashv$ which appends an endmarker\footnote{This function accounts for the fact that a one-way transducer (contrary to a Mealy machine) may perform some computation and produce some output at the end of the input word.},
 and (b) equivariant homomorphisms over polynomial orbit-finite alphabets that are not necessarily length-preserving. 
\section{Two-way single-use transducers}
\label{sec:two-way-transducers}

In this section, we turn to two-way single-use transducers. For them, we show three  other equivalent models: (a) compositions of certain two-way prime functions; (b)  an atom variant of the streaming string transducer (\sst) model of Alur and {\v C}ern{\'y} from~\cite{alurExpressivenessStreamingString2010}; and (c) an atom  variant of the regular string functions from~\cite{bojanczykRegularFirstOrderList2018}. We believe that the atom variants of items (b) and (c), as described in this section,  are the natural atom extensions of the original models; and the fact that these extensions are all equivalent to single-use two-way transducers is a further validation of the single-use restriction.

We  illustrate the  transducer models using the functions from the following example.
\begin{myexample}\label{ex:map-reverse-duplicate} Consider some polynomial orbit-finite alphabet $\Sigma$. 
    The input and output alphabets are the same, namely  $\Sigma$ extended with a separator symbol $|$. Define  \emph{map reverse} (respectively, \emph{map duplicate}) to be the function which reverses (respectively, duplicates) every string between consecutive separators, as in the following examples:
    \begin{align*}
    \underbrace{12||345|678|9 \ \mapsto\  21||543|876|9}_{\text{map reverse}}
    \qquad
    \underbrace{12||345|678|9 \ \mapsto\  1212||345345|678678|99}_{\text{map duplicate}}
    \end{align*}
    Both functions can be computed by single-use two-way transducers. These functions will be included in the prime functions for two-way single-use transducers, as discussed in item (a) at the beginning of this section.
\end{myexample}

\smallparagraph{Streaming string transducers with atoms} 
A streaming string transducer with atoms has two types of registers: atom registers $r,s,\ldots$ which are the same as in Definition~\ref{def:the-transducer-model}, and  string registers $A,B,C,\ldots$ which are used to store strings over the output alphabet. Both kinds of registers are subject to the single-use restriction, which is coloured red in the following definition. 

\begin{definition}
    [Streaming string transducer with atoms] \label{def:sst} Define the syntax of a \emph{streaming string transducer (\sst) with atoms} in the same way as a one-way single-use transducer (variant of Definition~\ref{def:the-transducer-model}), except that the model is additionally equipped with a finite set of \emph{string registers}, with a designated \emph{output string register}. The actions are the same as for one-way single-use transducers except that the output action is replaced by two kinds of actions:
    \begin{enumerate}
        \item \label{action:output-register-set} Apply  an equivariant function $f : \atoms^k \to \Gamma$ to the contents of distinct registers $r_1,\ldots,r_k \in R$, and put the result into  string register $A$ (overwriting its previous contents).  If any of these registers is undefined, then the run immediately stops and rejects. \red{This
        action has the side effect of setting the values of $r_1, r_2, \ldots, r_{k_j}$ to~$\bot$.}
        \item \label{action:output-register-concatenate} Concatenate string registers $A$ and $B$, and put the result into string register $C$. \red{This action has the side effect of setting $A$ and $B$ to the empty string.}
    \end{enumerate} 
\end{definition}

The output of a streaming string transducer is defined to be the contents of the designated output register when the ``accept'' action is performed.
In the atomless case, when  no atom registers are allowed  and the input and output alphabets are finite, the above definition is equivalent to the original definition of streaming string transducers from~\cite{alurExpressivenessStreamingString2010}.

\begin{myexample}\label{ex:iterated-reverse-sst} Consider the map reverse function from Example~\ref{ex:map-reverse-duplicate}, with alphabet $\atoms$.  To compute it, we use two string registers $A$ and $B$, with the output register being $B$. 
     When reading an atom $a \in \atoms$, the transducer executes an action $A:=aA$. (This action needs to be broken into simpler actions as in Definition~\ref{def:sst} and requires auxiliary registers). When reading a separator symbol, the automaton executes action $B:=B|A$, which erases the content of register $A$.  Similar idea works for map duplicate -- it uses two copies of register $A$.
\end{myexample}

\smallparagraph{Regular list functions with atoms} Our last model is based on the regular list functions from~\cite{bojanczykRegularFirstOrderList2018}. Originally, the regular list functions were introduced to  characterise  two-way transducers (over finite alphabets), in terms of simple prime functions and combinators~\cite[Theorem 6.1]{bojanczykRegularFirstOrderList2018}. The following definition extends the original definition\footnote{In~\cite{bojanczykRegularFirstOrderList2018}, the group product operation has output type $G^*$, while this paper uses   $(G \times \sigma)^*$. This difference is due to an error in~\cite{bojanczykRegularFirstOrderList2018}.} in only two ways:  we add an extra datatype $\atoms$ and an  equality test $\mathtt{eq} : \atoms^2 \to \set{\text{yes,no}}$.

\begin{definition}[Regular list functions with atoms] \label{def:reglist-fun} Define the  \emph{datatypes} to be sets which can be obtained from $\atoms$ and  singleton sets,  by applying constructors for products $\tau \times \sigma$, co-products $\tau + \sigma$ and  lists $\tau^*$.  The class of  \emph{regular list functions with atoms} is the least class  which:
    \begin{enumerate}
        \item \label{regfun:constant} contains all equivariant constant functions;
        \item  \label{regfun:primes} contains all  functions from Figure~\ref{fig:atomic-combinators-1}, and an equality test $\mathtt{eq} : \atoms^2 \to \set{\text{yes,no}}$;
        \item \label{regfun:combinators} is closed under applying the following combinators:
        \begin{enumerate}
            \item {\tt comp} function composition $(f,g) \mapsto f \circ g$;
            \item {\tt pair} function pairing $(f_0,f_1) \mapsto (x \mapsto (f_0(x),f_1(x)))$;
            \item {\tt cases} \label{it:cases} function co-pairing $(f_0,f_1) \mapsto ((i,a) \mapsto f_i(a))$;
            \item {\tt map} lifting functions to lists $f \mapsto ([a_1,\ldots,a_n] \mapsto [f(a_1),\ldots,f(a_n)])$.
        \end{enumerate}
    \end{enumerate}
\end{definition}

Every polynomial orbit-finite set is a datatype (actually, polynomial orbit-finite sets are exactly the datatypes that do not use lists), and therefore it makes sense to talk about regular list functions with atoms that describe string-to-string functions with input and output alphabets that are polynomial orbit-finite sets. Also, one can consider string-to-boolean functions -- they describe languages, and are the model mentioned in item~\ref{languages:list-function} of Theorem~\ref{thm:one-way-two-way}.

\newcommand{\atomicfunction}[3]{
\mathtt{#1} &:& #2\\
&& \red{\text{\footnotesize #3}}\\
}

\begin{figure}\footnotesize
    \begin{eqnarray*}
    \atomicfunction{project_i}{(\sigma_0 \times \sigma_1) \to \sigma_i}{projection $(a_0,a_1) \mapsto a_i$}
    \atomicfunction{coproject_i}{\sigma_i \to (\sigma_0 + \sigma_1)}{coprojection  $ a_i \mapsto (i,a_i)$}
    \atomicfunction{distr}{(\sigma_1 + \sigma_2) \times \tau \to (\sigma_1 \times \tau) + (\sigma_2 \times \tau)}{  distribution $((i,a),b) \mapsto (i,(a,b))$}
    \atomicfunction{reverse}{\sigma^* \to \sigma^*}{list reverse $[a_1,\ldots,a_n] \mapsto [a_n,\ldots,a_1]$}
    \atomicfunction{concat}{(\sigma^*)^* \to \sigma^*}{list concatenation, defined by $[] \mapsto []$ and   $[a] \cdot l  \mapsto a \cdot concat(l)$}
    \atomicfunction{append}{(\sigma \times \sigma^*) \to \sigma^*}{append, defined by $(a,l) \mapsto [a] \cdot l$}
    \atomicfunction{coappend}{\sigma \to (\sigma \times \sigma^*)  + \bot}{the opposite of append, defined by $[] \mapsto (1,\bot)$ and $[a] \cdot l \mapsto (0,(a,l))$}
    \atomicfunction{block}{(\sigma+\tau)^* \to (\sigma^* + \tau^*)^*}{group the list into maximal connected blocks from $\sigma^*$ or $\tau^*$}
    \atomicfunction{group}{(G \times \sigma)^* \to (G \times \sigma)^*}{$[(g_1,a_1),\ldots,(g_n,a_n)] \mapsto [(1,a_1),(g_1,a_2),(g_1g_2,a_3),\ldots,(g_1 \cdots g_{n-1},a_n)]$ }
    \end{eqnarray*}
      \caption{\label{fig:atomic-combinators-1} For every datatypes $\tau, \tau_0, \tau_1, \sigma$, every finite group $G$, and every $i \in \set{0,1}$  the above functions are regular list functions with atoms.}
    \end{figure}

\begin{example}\label{ex:map-reverse}
    We show that map reverse  from Example~\ref{ex:map-reverse-duplicate} is a regular list function with atoms. Consider an input string, say
    \begin{align*}
        [1,2,|,|,3,4,5,|,6,7,8,|,9] \in (\atoms + |)^*.
    \end{align*} 
    Apply the prime {\tt block} function, yielding 
    \begin{align*}
        [[1,2],[|,|],[3,4,5],[|],[6,7,8],[|],[9]]  \in (\atoms^* + |^*)^*.
    \end{align*} 
    Using the {\tt cases} and {\tt map} combinators, apply {\tt reverse} to all list items, yielding 
    \begin{align*}
        [[2,1],[|,|],[5,4,3],[|],[8,7,6],[|],[9]]  \in (\atoms^* + |^*)^*.
    \end{align*} 
    To get the final output, apply {\tt concat}. A similar idea works for map duplicate, except we need to derive the string duplication function:
    \begin{align*}
    w  \stackrel {{\tt pair}(\ldots)}\mapsto
    (w,[w]) \stackrel{\tt append}{\mapsto}
    [w,w] \stackrel{\tt concat}{\mapsto}
    ww
    \end{align*}
\end{example}

\smallparagraph{Equivalence of the models} The main result of this section is that all models described above are equivalent, and furthermore admit a decomposition into prime functions in the spirit of the Krohn-Rhodes theorem. Since the functions discussed in this section are no longer length-preserving, the Krohn-Rhodes decomposition uses only sequential composition.

\begin{theorem}\label{thm:two-way-models} The following conditions are equivalent for every total function $f : \Sigma^* \to \Gamma^*$, where $\Sigma$ and $\Gamma$ are polynomial orbit-finite sets:
    \begin{enumerate}
        \item \label{two-way:two-way} $f$ is computed by a two-way single-use transducer;
        \item \label{two-way:sst} $f$ is computed by a streaming string transducer with atoms;
        \item \label{two-way:regular-list-function} $f$ is a regular list function with atoms;
        \item \label{two-way:krohn-rhodes} $f$ is a sequential composition of  functions of the following kinds:
        \begin{enumerate}
            \item single-use Mealy machines;
            \item equivariant homomorphisms that are not necessarily length-preserving;
            \item map reverse and map duplicate functions from Example~\ref{ex:map-reverse-duplicate}.
        \end{enumerate}
    \end{enumerate}
\end{theorem}
In the future, we plan to extend the above theorem with one more item, namely a variant of \mso transductions based on  \rgmso. 
The models in items~\ref{two-way:regular-list-function} and~\ref{two-way:krohn-rhodes} are closed under sequential composition, and therefore the same is true for the models in items~\ref{two-way:two-way} and~\ref{two-way:sst}; we do not know any direct proof of composition closure for items~\ref{two-way:two-way} and~\ref{two-way:sst}, which contrasts the classical case without atoms~\cite[Theorem 2]{chytilSerialComposition2Way1977}.
 The Krohn-Rhodes decomposition from item~\ref{two-way:krohn-rhodes}, in the case without atoms, was present implicitly in~\cite{bojanczykRegularFirstOrderList2018}; in this paper we make the decomposition explicit, extend it to atoms, and  leverage it to get a relatively simple proof of Theorem~\ref{thm:two-way-models}. Even  for the reader interested in transducers but not  atoms, our Krohn-Rhodes-based proof of Theorem~\ref{thm:two-way-models} might be of some independent interest.

 Here are some immediate  corollaries of Theorem~\ref{thm:two-way-models}:
\begin{enumerate}
    \item Every function in item~\ref{two-way:krohn-rhodes} is computed by a two-way single-use transducer  which is  reversible in the sense of~\cite[p.~2]{DartoisFJL17}; hence two-way single-use transducers can be translated into reversible ones.
    \item Since the equivalence in Theorem~\ref{thm:two-way-models} also works for functions with yes/no outputs, it follows that items~\ref{languages:two-way} and~\ref{languages:list-function} in Theorem~\ref{thm:one-way-two-way} are equivalent.
    \item If $f$ is a transducer from the class described in  Theorem~\ref{thm:two-way-models}, then the language class described  in Theorem~\ref{thm:one-way-two-way} is preserved under inverse images of $f$. 
\end{enumerate}

All conversions between the models in Theorem~\ref{thm:two-way-models} are effective. 
Our last result concerns the equivalence problem for these models, which is checking if two transducers compute the same function. Using a reduction to the equivalence problem for \emph{copyful} streaming string transducers without atoms~\cite[p.~81]{filiotCopyfulStreamingString2017}, we prove the following result:
\begin{theorem}\label{thm:equivalence}
    Equivalence  is decidable for streaming string transducers with atoms (and therefore also for every other of the equivalent models from Theorem~\ref{thm:two-way-models}).
\end{theorem}

\bibstyle{alpha}
\bibliography{bib}
 
\vfill
\pagebreak
\appendix
\section{Omitted Lemma from Section~\ref{sec:automata-to-semigroups}}
\label{sec:functions-with-limited-support}
In this section we prove the following lemma that was used in
Section \ref{sec:automata-to-semigroups}:
\begin{lemma}

    Let $X, Y$ be orbit-finite sets and let $F \subseteq X \to Y$ be a equivariant set of
    finitely supported functions. Then $F$ is orbit finite if and only if there is a limit $k \in \mathbb{N}$,
    such that every function in $F$ is supported by at most $k$ atoms.
\end{lemma}
\begin{proof}
    For the left-to right implication,  observe that all the functions in one orbit can be supported
    by the same number of atoms. To finish the proof of the implication, set $k$ to be the largest of those numbers.\\
    Consider now the right-to-left implication. Choose some $\bar a \in \atoms^k$. There are finitely many
    functions supported by $\bar a$, because every such function is a union of $\bar a$-orbits
    of $X \times Y$ and there are only finitely many such orbits (\cite[Theorem 3.16]{bojanczyk_slightly2018}).
    Every function $f : X \to Y$ with a support of size at most $k$,
    can be obtained by applying a suitable atom automorphism to function $f': X \to Y$ that is
    supported by $\bar a$.
\end{proof}

\section{Factorisation forests}
\label{sec:factorisation-forests}
In this part of the appendix, we prove a variant of the Factorisation Forest Theorem of Imre Simon~\cite[Theorem 6.1]{simonFactorizationForestsFinite1990}.
This result will be used three times in this paper: (1) to construct the Krohn-Rhodes decomposition for single-use Mealy machines from Theorem~\ref{thm:mealy-machine-krohn-rhodes}; (2) to prove the implication from orbit-finite monoids to  single-use automata in Theorem~\ref{thm:one-way-two-way};
and (3) to prove the Krohn-Rhodes decomposition for two-way single-use
transducers in Theorem~\ref{thm:two-way-models}.

The general idea  is  the same as in the Factorisation Forest Theorem: given a monoid homomorphism, one associates to each string a tree structure on the string positions, so that (a)  sibling subtrees are similar  with respect to the homomorphism; and (b) the depth of the trees is bounded.  
 To represent tree decompositions, we use Colcombet's {splits}~\cite[Section 2.3]{colcombetCombinatorialTheoremTrees2007}, because they are easily handled by one-way deterministic transducers.

Define a \emph{split} of a string to be a function   which assigns to each string position  a  number from $\set{1,2,\ldots}$, called the \emph{height} of the position. 
 A split induces the following tree structure on string positions. Let $i,j$ be two positions.  We say that $i$ is  a \emph{descendant} of  $j$  (with respect to a given split) if $i \le j$ and  all positions in the interval between $i$ and $j$ (including $i$ but not including $j$) have heights strictly smaller than
 the height of $j$. The descendant relation defined this way is a tree ordering: it is transitive, reflexive, anti-symmetric, and the ancestors  (i.e.~the opposite of descendants) of a string position  form a chain.  We say that $i$ is a \emph{sibling} of $j$ if both positions have the same height (say $k$) and all positions in the interval between $i$ and $j$ (it does not matter if the endpoints are included) have heights not greater than $k$. The sibling relation -- unlike the descendant relation -- is symmetric. Here is an example split:
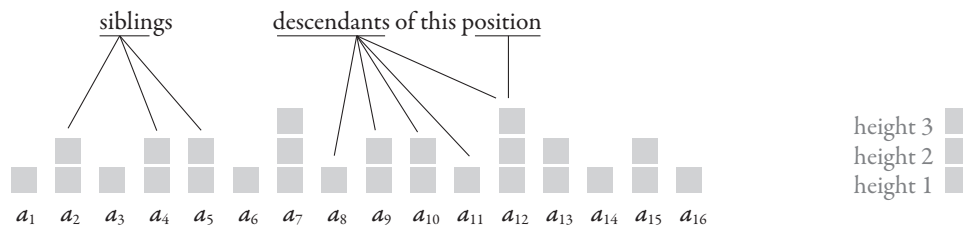
\mypic{11}
The tree structure induced by a split will
usually be a forest, with several roots, rather than a tree. Note also that the sibling relation as defined above is not the same thing as having the same parent (nearest ancestor): for example in the above picture $a_{10}$ and $a_{11}$ have the same parent but are not siblings, because they have different heights.

Suppose that $h : \Sigma^* \to M$ is a monoid homomorphism, and $\sigma$ is a split of a string $w \in \Sigma^*$. For a position $i$ in $w$, define its \emph{split value}  to be the value under $h$ of the interval in $w$ consisting of the  descendants of $i$. Simon's  original Factorisation Forest Theorem,
when expressed in the terminology of splits, says that if $M$ is finite (not just orbit-finite), then there is a split with  height bounded by a function of $M$ such that: ($\dagger$) in every set of at least two siblings, the split values are all equal and furthermore idempotent
\footnote{A monoid element $e$ is an idempotent if $ee = e$.}. Still  assuming that $M$ is finite, Colcombet shows that splits can be produced by Mealy machines, under a certain relaxation of condition ($\dagger$). In this paper, we show that if the monoid is orbit-finite, and condition ($\dagger$) is further relaxed,  then splits can be produced by compositions of prime functions, as in Theorem~\ref{thm:mealy-machine-krohn-rhodes}.

Our relaxation of ($\dagger$) is defined in terms of smooth sequences~\cite[Section 3.1]{kirsten2004}.
We say that a sequence $x_1, \ldots, x_n$ of monoid
elements is  \emph{smooth},  if
for every $i \in \set{1,\ldots,n}$ there exist monoid elements $y,z$ such that the product $y x_1 \cdots x_n z$ is equal to $x_i$. In the terminology of Green's relations -- see Section~\ref{sec:greens-relations} -- this means that all of the monoid elements $x_1,\ldots,x_n$ are $\Jj$-equivalent, and they are also $\Jj$-equivalent to the product $x_1 \cdots x_n$.   Thanks the local theory of finite monoids, smooth sequences have a strong structure, related to the Rees-Shushkevich decomposition, see~\cite[Proposition 4.36]{PinMPRI}.
This structure carries over to orbit-finite monoids, and will be used in the proofs.

We say that a split of a string in $\Sigma^*$ is  \emph{smooth} with respect to a homomorphism $h : \Sigma^* \to M$  if for every set of consecutive siblings $i_1 < \ldots < i_k$, the corresponding sequence of split values is smooth.
The goal of this section is to prove the Split Lemma, which says that for every orbit-finite monoid $M$, every polynomial orbit-finite alphabet $\Sigma$
and every equivariant monoid morphism $h: \Sigma^* \to M$, there is a composition of primes which maps each string in $\Sigma^*$ to a smooth split. 

In the Split Lemma, we use a technical assumption about least supports. 
A non-repeating tuple of atoms $\bar a \in \atoms^*$ is called a {\em least support} of a finitely-supported $x$
    if (a) $\bar a$~supports $x$ and (b) every atom from $\bar a$ appears in every other tuple of atoms that supports $x$. The least support is unique, up to reordering the atoms in the tuple.  Every finitely-supported element has a least
    support~\cite[Theorem 6.1]{bojanczyk_slightly2018}.

\begin{lemma}[Split Lemma]\label{lem:split-lemma}
    Let $\Sigma$ be a polynomial orbit-finite set,
    $M$ an orbit-finite monoid, and let $h : \Sigma^* \to M$ be an equivariant monoid homomorphism such that 
    \begin{itemize}
        \item[(*)] for every $x \in M$ there is some letter $a \in \Sigma$, called the \emph{letter representation of $x$}, such that $h(a)=x$ and $a$ has the same least support as $x$. 
    \end{itemize}
There exists a bound on the height  $n \in \set{1,2,\ldots}$ and a composition of primes
\begin{align*}
f : \Sigma^* \to 
(
    \overbrace{\set{1,\ldots,n}}^{\substack{
        \text{Position's} \\ 
        \text{height}
        }}\times 
    \overbrace{\Sigma}^{\substack{
        \text
        {Letter representation of the $h$-image}\\
        \text{of the subword of position's descendants}
    }}
)^*
\end{align*}
such that for every input string, the output of  $f$ represents a smooth split.
    \mypic{12}
\end{lemma}

The assumption (*) in the lemma is technical, and is used to overcome the fact that compositions of primes have output alphabets which are polynomial orbit-finite sets, while the monoid $M$ might be orbit-finite but not polynomial. The following lemma shows that any homomorphism can be extended to satisfy (*).

\begin{lemma}
    \label{lem:extend-to-star} Let $h : \Sigma^* \to M$ be an equivariant monoid homomorphism, with $\Sigma$ polynomial orbit-finite and $M$ orbit-finite. There is a polynomial orbit-finite set $\Gamma \supseteq \Sigma$ and an extension $g : \Gamma^* \to M$ of $h$ such that $g$ satisfies condition (*) in the \SplitLemma.
\end{lemma}
\begin{proof}
    By~\cite[Theorem 6.3]{bojanczyk_slightly2018}, there is a set $\Delta$ and a surjective function $f : \Delta \to M$ such that $f$ preserves and reflects supports. Define $\Gamma = \Sigma + \Delta$ and define $g$ to be the unique extension of $h$ which coincides with $f$ on elements of $\Delta$.
\end{proof}
The rest of Appendix~\ref{sec:factorisation-forests} is dedicated to proving the Split Lemma. In Section~\ref{sec:uniformisation}, we discuss questions of choice and uniformisation for  polynomial orbit-finite sets, as well as elimination of spurious atom constants. In Section~\ref{sec:greens-relations},  we recall basic results about   Green's relations for orbit-finite monoids that were proved in~\cite{bojanczykNominalMonoids2013}. In Section~\ref{sec:closure-properties}, we prove basic  closure
properties of compositions  of primes.  We conclude by proving the 
\SplitLemma.

\subsection{Uniformisation and elimination of spurious constants}
\label{sec:uniformisation}
In the proof of \SplitLemma, we will use
some choice constructions, e.g.~choosing an element from each $\Jj$-class (see below). This could be problematic, 
because one of the difficulties with orbit-finite sets is that axiom of choice can fail, as illustrated by the following example.
\begin{example}\label{ex:choice}
    The surjective function $(a,b) \mapsto \set{a,b}$, which maps an ordered pair of atoms to the corresponding unordered pair, has no finitely supported one-sided inverse
    (\cite[Example 9]{bojanczyk_slightly2018}). In other words,  there is no  finitely supported way of choosing one atom from an unordered pair of atoms. 
\end{example}
The difficulties with choice illustrated in Example~\ref{ex:choice} stem from symmetries like $\set{a,b}= \set{b,a}$.
We do not encounter such problems with polynomial orbit-finite sets, because they only use ordered tuples of atoms 
-- we know which atom stands on the first position, which atom stands on the second position, and so on.
Lemma~\ref{lem:write-support} below uses this to extract canonical supports of elements of polynomial orbit-finite sets.

\begin{lemma}
    \label{lem:write-support}
    Let $\Sigma$ be a polynomial orbit-finite set and let  $k \in \set{0,1,\ldots}$ be the maximal size of least supports in $\Sigma$ (we call this the \emph{dimension} of $\Sigma$). There is an equivariant function $\textrm{supp}: \Sigma \to \atoms^{\leq k}$, which maps every element of $\Sigma$ to a tuple that is its least support.
\end{lemma}
\begin{proof}
Induction on the construction of $\Sigma$.
\end{proof}

A key  benefit of polynomial orbit-finite sets is that some problems with choice can be avoided, namely from each binary relation one can extract a function, which is called its uniformisation. Uniformisation will play a crucial role in our construction. From this perspective, the decision to use polynomial orbit-finite sets, instead of general orbit-finite sets, is important.

\begin{lemma}[Uniformisation]\label{lem:straight-choice} 
    Let $R \subseteq  X \times Y$ be a finitely supported
    binary relation on polynomial orbit-finite sets, such that for every 
    $x \in X$, its image under $R$: 
    \[x R = \set{y~:~R(x, y)}\] 
    is nonempty. Then $R$ can be uniformised, i.e.~there is a finitely supported function
    \[f~:~X~\to~Y\] 
    such that all  $x \in  X$ satisfy $x \, R \, f(x)$.
\end{lemma}
\begin{proof}
    Let $\bar a$ be the least support of $R$. 
    Choose a tuple $\bar c$ of atoms which contains $\bar a$ plus $2n$ fresh atoms, where $n$ is the dimension of $Y$. This tuple will support $f$.
\begin{claim} 
    For every $x \in X$, there is some $y \in x R$ such that every atom from the least support of $y$
    appears in the least support of $x$ or in $\bar c$.
\end{claim}
\begin{proof}
        Let $\bar b = \supp{x}$. 
        The function $x \mapsto x R$ is supported by $\bar a$. 
        It follows that the set $x R$ is supported by the tuple $\bar {a b}$.
        Choose some  $y \in x R$. Because $x R$ is supported by $\bar {ab}$,
        every atom from $y$ that does not appear in $\bar {a b}$ can be replaced by another
        atom that does not appear in $\bar {a b}$, as long as the equality type is preserved.
        Since $\bar c$ was chosen to be big enough, we can assume without loss of generality that every atom in $y$ appears either in the least support of $x$ or in $\bar c$.  
\end{proof}
For $x \in X$ (and $\bar b = \supp{x}$),  consider the following order on the elements of $x R$ that are supported by $\bar {cb}$:
(a) atoms are ordered by their position in $\bar {cb}$ (b) tuples are ordered lexicographically
(c) elements of a co-product are ordered ``left before right''.
The function which maps $x$ to the least element from $x R$ according to the above ordering is easily seen to be supported by $\bar c$, and it is the uniformisation required by the claim. 
\end{proof}

\begin{myexample}\label{ex:uniformisation}
    Let $X= \atoms$, $Y = \atoms^2$, and let $R \subseteq X \times Y$ be the inverse of the projection $(a,b)  \mapsto a$. A uniformisation of $R$ is the function $a \mapsto (a,1)$, where $1$ is an atom constant. This uniformisation is not equivariant, because it uses the constant $1 \in \atoms$, and there is no uniformisation that is equivariant.
\end{myexample}

The choice of the tuple $\bar c$ in the proof of  Lemma \ref{lem:straight-choice} raises its own problems, because it makes the construction non-equivariant in the following sense: even if the relation $R$ is equivariant, the choice  function $f$ might not be equivariant (this is necessary as witnessed by Example~\ref{ex:uniformisation}). 
Because of that, even when decomposing an equivariant function $f$
into primes, it will be useful to use homomorphisms that are not equivariant.
This motivates the following definition. 
\begin{definition}
    For a tuple of atoms $\bar a \in \atoms^*$, we say that a function $f$ is a composition of $\bar a$-primes if it can be obtained using parallel and sequential  compositions from
    classical Mealy machines, atom propagation, and $\bar a$-supported homomorphisms.
\end{definition}
To solve the problem with using non-equivariant functions that arises from an application of Lemma~\ref{lem:straight-choice}, we prove below  that when defining equivariant functions,
compositions of (equivariant) primes
have exactly the same expressive power as compositions of $\bar a$-primes (for every $\bar a$).
\begin{lemma}[Elimination of spurious constants]
\label{lem:superflous-atom-elimination}
If an equivariant function is a composition of $\bar a$-primes, for some tuple of atoms $\bar a$, then  it is also a composition of (equivariant) primes.
\end{lemma}
\begin{proof}
    Let $f$ be an equivariant function. 
    \begin{claim}
        If $f$ is a composition of $\bar a$-primes, then it is also a composition
        of $\bar b$-primes, for every tuple $\bar b$ of the same size as $\bar a$
        (assuming that neither $\bar a$, nor $\bar b$ contains repeating atoms).
    \end{claim}
    \begin{proof}
        Take an atom automorphism $\pi$ that transforms $\bar a $ to $\bar b$. 
        Both the parallel and the sequential compositions are equivariant (higher order) functions, so we can just apply
        $\pi$ to every function in a derivation of $f$ as composition of $\bar a$-primes, to obtain a
        derivation of $\pi(f)$ as a composition $\pi(\bar a)$-primes. To finish the proof,
        notice that $\pi(f) = f$ by equivariance of $f$ and  $\pi(\bar a) = \bar b$ by definition of $\pi$.
    \end{proof}
    Thanks to the claim above, it would suffice to choose such $\bar b$ that would only contain
    equivariant atoms (fresh for everything). Since there are no such atoms, we introduce the concept
    of {\em placeholders} instead\footnote{The concept of placeholders is a special case of the  {\em name abstraction}
    from \cite[Section 4]{pitts2013nominal}.}.
    \begin{definition}
        \label{definition:placeholders}
        Fix a finite set $P$, not containing any atoms, whose elements will be called  placeholders.
        For every polynomial orbit-finite set $\Sigma$, define the set $\Sigma(P)$ inductively
        on the structure of $\Sigma$:
        \begin{align*}
        \overbrace{\atoms(P)= \atoms + P}^{\text{atoms}} \ \ 
        \overbrace{\set{x}(P) = \set{x}}^{\text{singletons}}  \ \ 
        \overbrace{(\Sigma \times \Gamma)(P)=\Sigma(P) \times \Gamma(P)}^{\text{products}} \ \ 
        \overbrace{(\Sigma+\Gamma)(P)=\Sigma(P) + \Gamma(P)}^{\text{co-products}}
        \end{align*}
    \end{definition}
    
       There is a natural equivariant embedding $\mathsf{cast}_P : \Sigma \to \Sigma(P)$.
        It has a one-sided inverse
        $\mathsf{cast}_P^{-1} : \Sigma(P) \to \Sigma + \bot$, which only
        works for values that do not contain elements~from~$P$.

        Intuitively, the set $\Sigma(P)$ is the set $\Sigma$, in which values from $P$
        behave like atoms. This means that $\Sigma(P)$ is equipped with an action
        of atom-and-placeholder automorphisms ($\atoms + P \to \atoms + P$) (in addition to the 
        natural action of atom automorphisms $\atoms \to \atoms$). The atom-and-placeholder action can
        be used, for example, to substitute a particular atom with a placeholder. This action 
        extends to functions $\Sigma(P) \to \Gamma (P)$.
        Moreover, every finitely supported function in $\Sigma \to \Gamma$ can be lifted
        to $\Sigma(P) \to \Gamma(P)$:
        \begin{claim}
        \label{claim:function-extend-to-placeholders}
        For every polynomial orbit-finite sets $\Sigma$ and $\Gamma$ and
        every finitely-supported function $f : \Sigma \to \Gamma$, there
        is a function $f' : \Sigma(P)  \to \Gamma(P)$ such that 
        the following diagram commutes
        \begin{eqnarray*}
        \xymatrix{
        \Sigma \ar[d]_{\pi \circ \mathsf{cast}_P}  \ar[r]^{f} & \Gamma  \\
        \Sigma(P) \ar[r]_{f'} & \Gamma(P) \ar[u]_{\mathsf{cast}_P^{-1} \circ \pi^{-1}}
        }
        \end{eqnarray*}
        for every atom-and-placeholder automorphism $\pi$ that does
        not touch the atoms
        in the least support of $\bar f$.
        \end{claim}
        \begin{proof}
            We obtain $f'$ by treating $f$ as a subset of $\Sigma \times \Gamma$, casting it
            to $\Sigma(P) \times \Gamma(P)$, and closing it under atom-and-placeholder
            $\supp f$-automorphisms.
        \end{proof}
    
    To finish the proof of Lemma~\ref{lem:superflous-atom-elimination}, choose $P$ of size
    at least $|\bar a|$, and construct $f$ as the following composition:
    (a) cast all the input letters into their placeholder versions (equivariant homomorphism -- $\mathsf{cast}_P$),
    (b) apply $f'$ constructed as composition of $\bar b$-primes, where $\bar b \in P^*$
    (equivariant primes) and (c) cast all the letters back to
    the original alphabet (equivariant homomorphism -- $\mathsf{cast}_P^{-1}$).
    Inputs of $f'$ contain no placeholders and $f$ is equivariant. 
    It follows that outputs of $f'$ also contain no placeholders, which proves that the transformation
    (c) is always well defined.
\end{proof}
From now on, we  will freely use finitely supported
homomorphisms (rather then only equivariant ones)
when proving that an equivariant function is a
composition of primes.

\subsection{Green's relations}
\label{sec:greens-relations}
In this subsection we recall some basic facts about
Green's relations for orbit-finite monoids. These will be used in the proof of the  \SplitLemma, and later on
in the paper. 

\begin{definition}
    [Green's relations]
    For elements $x$ and $y$ in  a monoid, say that  $x$ is an {\em infix} of $y$, if $y=axb$ for some $a$, $b$ from the monoid. Similarly, we say that $x$ is a prefix of $y$ if $y = xa$ for some $a$, and that $x$ is a suffix of $y$
if $y = ax$ for some $a$. We say that $x$ and $y$ are $\Jj$-equivalent if they are each other's infixes. Equivalence classes of $\Jj$-equivalence are called  $\Jj$-classes. In the same way we define $\Rr$-classes for prefixes
and $\Ll$-classes for suffixes.
\end{definition}

The following orbit-finite adaptation of the classical Eggbox Lemma for finite monoids was proved in~\cite[Lemma 7.1 (first item) and Theorem 5.1]{bojanczykNominalMonoids2013}.
\begin{lemma} [Eggbox Lemma]
    \label{lem:green-eggbox} 
    Let $x$ and $y$ be elements of an  orbit-finite monoid. If $x$ and $xy$ are in the same $\Jj$-class, then they are also in the same
    $\Rr$-class (prefix class). Similarly if $y$ and $xy$ are in the same $\Jj$-class, then they are also in the same
    $\Ll$-class (suffix class). 
\end{lemma}

In the terminology of Green's relations,  a sequence of monoid elements $x_1, \dots, x_n$ is smooth if all the list elements $x_1,\ldots,x_n$
as well as the monoid product $x_1 \cdots x_n$ are $\Jj$-equivalent. An important corollary of the Eggbox Lemma is the locality  of smooth sequences.

 \begin{lemma}[Locality of smoothness]\label{lem:smooth-local}
    A sequence $x_1, \ldots, x_n$ is smooth if and only if $x_i, x_{i + 1}$ is smooth for every $i < n$.
 \end{lemma}
 \begin{proof}
    An
    induction strategy reduces the proof of the lemma to the case of $n=3$, i.e.~showing that $x_1, x_2, x_3$ is smooth if and only if
    both $x_1,x_2$ and $x_2,x_3$ are smooth. The left-to-right implication is immediate, the right-to-left implication
    follows from the Eggbox  Lemma.
 \end{proof}

 We finish this subsection, by defining a way of measuring the size of orbit-finite monoids. Define {\em $\Jj$-height} of an orbit-finite monoid  to be the length of the longest strictly increasing sequence with respect to the infix relation.
 
 \begin{lemma}
    \label{lem:j-height-finite}
    The $\Jj$-height of every orbit-finite monoid is finite\footnote{
        In this paper, we  study atoms with equality only.
        Although orbit-finite monoids make sense also for atoms with more structure than just equality, Lemma~\ref{lem:j-height-finite} can fail for such atoms. For example, if the atoms are the rational numbers equipped with $\leq$, then the monoid of atoms where the product operation is taking the maximum out of two numbers (equipped with an identity element $-\infty$) is orbit-finite, but its $\Jj$-height is infinite. This is one of the
         problems that has to be overcome when extending the results of this paper to atoms with more structure.
    }.
 \end{lemma}
 \begin{proof}
    By~\cite[last line of proof of Lemma 9.3]{bojanczykNominalMonoids2013} if two elements of an orbit-finite monoid are in the same orbit, then they are either in the same $\Jj$-class, or are incomparable with respect to the infix relation.
    The lemma follows.
 \end{proof}

 \subsection{Closure properties}
 \label{sec:closure-properties}
In this subsection, we note that the compositions of primes are closed under certain simple combinators. For a string-to-string function $f : \Sigma^* \to \Gamma^*$,  and an alphabet $\Delta$ of separators (all $\Sigma, \Gamma, \Delta$ polynomial orbit-finite), define
\begin{align*}
\map_\Delta f, \sub_\Delta f : (\Sigma + \Delta)^* \to (\Gamma + \Delta)^*,
\end{align*}
as follows. The first function applies $f$ separately to each interval between separators:
\begin{align*}
    w_0 a_1 w_1 \cdots a_n w_n \quad \stackrel{\map_\Delta f}\mapsto \quad f(w_0) a_1 f(w_1)  \cdots a_n f(w_n) 
    \end{align*}
    where $w_0,\ldots,w_n \in \Sigma^*$ and $a_1,\ldots,a_n \in \Delta$. 
The second function applies $f$ to the entire string (ignoring the separators), and keeps the separators in their original positions (this definition only makes sense for a length-preserving $f$). If we omit the subscript $\Delta$ from $\map_\Delta$ or $\sub_\Delta$, then we assume that $\Delta$ has a single separator symbol denoted by $|$.

\begin{myexample}
    Consider the length-preserving function 
    \begin{align*}
    f : \set{a,b}^* \to \set{a,b}^*   \qquad w \mapsto \sigma^{|w|} \text{ where $\sigma$ is the first letter of $w$}.
    \end{align*}
    For this function, we have  
    \begin{eqnarray*}
    aba|ba|bba  &\quad \stackrel{\map f}\mapsto \quad &  aaa|bb|bbb  
    \\
    aba|ba|bba  &\quad \stackrel{\sub f}\mapsto \quad &  aaa|aa|aaa 
    \end{eqnarray*}        
\end{myexample}

We also need a variant of map where the separators are included in the end of each block. More formally, for every function
\begin{align*}
f : (\Sigma + \Delta)^* \to \Gamma^*
\end{align*}
define $\map'_\Delta f : (\Sigma + \Delta)^* \to \Gamma^*$ in the following way:
\begin{align*}
    w_0 a_1 w_1 \cdots a_n w_n \quad \stackrel{\map' f}\mapsto \quad f(w_0 a_1) f(w_1 a_1)  \cdots f(w_{n-1}a_n) f(w_n) 
    \end{align*}
    where $w_0,\ldots,w_n \in \Sigma^*$ and $a_1,\ldots,a_n \in \Delta$. 

    \begin{lemma} \label{lem:prime-combinators}
    Compositions of primes are closed under the three combinators described above. 
\end{lemma}
\begin{proof}
 It suffices to see that
applying the combinators to each of the prime functions can be expressed as a composition of prime functions;
and that the combinators commute with compositions e.g:
$
    \mathtt{map} \; (f \circ g) = (\mathtt{map} \; f) \circ (\mathtt{map} \; g).
$
\end{proof}

 \subsection{Proof of the Split Lemma}
 We are now ready to prove the Split Lemma. Recall that the Split Lemma says that a composition of primes can compute smooth splits, represented as in the following picture
 \mypic{12}

\begin{remark}
    \label{rk:split-exclusive-values}
    In order to make construction in  Sections \ref{ap:mealy-primes} and \ref{sec:moniod-to-automaton} easier, it is useful to decorate
    every position by another auxiliary value
     -- a letter representation of the value under $h$ of the interval of the position's {\em strict} descendants, which
    are all the descendants of a position without
    the position itself. Computing these auxiliary values is no harder than computing the values in the Split Lemma, since the auxiliary values can be recovered by considering a homomorphism into the monoid $M' = M \times M$, where the product is 
    \begin{align*}
        (v_1, x_1) \cdot (v_2, x_2) = (v_1 \cdot v_2, v_1 \cdot x_2)
    \end{align*}
    and the homomorphism is defined by  $a \mapsto  (h(a), 1)$, where $1 \in M$ is the identity element.
\end{remark}
 
We prove the Split Lemma for the more general case of orbit-finite semigroups, rather than only for monoids. (Therefore, instead of $M$ we will write $S$ for the target of the homomorphism.)
This is useful, because when restricting to a proper subset of the semigroup we will not need to show that the subset contains an identity element.
The results about Green's relations from Section~\ref{sec:greens-relations} hold for orbit-finite semigroups as well.
The proof is by induction on the $\Jj$-height of the semigroup, which is finite by Lemma \ref{lem:j-height-finite}.

\smallparagraph{Induction base}
The induction base is when the $\Jj$-height of the semigroup is  equal to $1$. In this case, there is only one $\Jj$-class.
\begin{claim}
    \label{claim:height-one-smooth}
    If a semigroup  has $\Jj$-height equal to $1$,
    then it has only one $\Jj$-class.
\end{claim} 
\begin{proof}
    Consider elements  $x, y$ in the semigroup. Since $x, xy$ is a chain with respect to the infix ordering, and strictly increasing chains have length 1, it follows that $x$ and $xy$ are in the same $\Jj$-class. The same argument shows that $y$ and $xy$ are in the same $\Jj$-class, and therefore $x$ and $y$ are in the same $\Jj$-class.
\end{proof}
If there is only one $\Jj$-class, then every sequence is smooth. This means that we can obtain a smooth split,
by assigning $1$ as the height of every position,
and keeping the input
letters as the split values.
This is easily achievable by a homomorphism.

\smallparagraph{Induction step} We are left with the induction step. We begin with a lemma
about computing smooth products. When stating the lemma, we describe the input as two strings of equal length (first and second row); this is formalised by considering a single input string over a product alphabet. We say that a string $a_1 \cdots a_k \in \Sigma^*$ is smooth if the sequence $h(a_1),\ldots,h(a_k)$ is a smooth sequence in the semigroup.
 \begin{lemma}
    \label{lem:smooth-product-prime}
    There is a composition of primes which does the following:
    \[
        \begin{matrix}
            \mathtt{Input} & a_1 & a_2 & a_3 & \ldots & a_k \\
            \mathtt{Input} &     &     &     &        & \dashv \\
            \mathtt{Output}&     &     &     &        &  a \\
        \end{matrix}
    \]
    The inputs are: a smooth string  $a_1 \cdots a_k \in \Sigma^*$ in the first row, and a sequence of blanks followed by  an endmarker in the second row. The output is a sequence of blanks, followed by a letter $a \in \Sigma$
    such that $h(a)= h(a_1, \ldots ,a_k)$ (i.e
    a letter representation of $h(a_1 \cdots a_k)$).
    \end{lemma}

    In the lemma, there is no requirement on outputs for inputs which do not satisfy the assumptions stated in the lemma, i.e.~for inputs where the first row is not smooth or where the second row is not a sequence of blanks followed by an endmarker.
    \begin{proof}
        We compute the product in five steps described below. We assume
        that $k > 1$. If it is not, we can verify this by checking that the first letter contains the endmarker and simply copy the value $a_1$ to the output.
        \begin{enumerate}
            \item Let $J$ be the $\Jj$-class that contains all the semigroup elements $h(a_i)$ and
                  their product; such a $\Jj$-class exists by assumption on the first row being smooth.
                  The class $J$ contains an idempotent (recall that an idempotent is such an element $e \in S$, that $ee=e$) --  
                  this is because sequence $h(a_1), \ldots h(a_k)$
                  is smooth and $k \geq 2$, it follows that
                  $J \cap (J \cdot J)$ is nonempty, which implies
                  that $J$ contains an idempotent \cite[Corollary 2.25]{PinMPRI}.
                  In this step we decorate each position $i \in \set{1,\ldots,k}$ with a letter representation of
                  idempotent $e_i \in  J$ such that all the idempotents $e_1,\ldots,e_k$ have the same
                  least support.
                  This can be done by using a homomorphism to apply to
                  every letter the function $E$ from the following claim:
                  \begin{claim}
                    There is a finitely supported function $E : \Sigma \to \Sigma$ such that:
                    \begin{enumerate}
                        \item for every $a \in \Sigma$, $E(a)$ represents an idempotent from the $\Jj$-class of $h(a)$,
                              provided that one exists;
                        \item \label{it:same-least-support-idempotent} if $a_1, a_2 \in \Sigma$ represent elements from
                              the same $\Jj$-class, then $E(a_1)$ and  $E(a_2)$ have the same least supports.
                    \end{enumerate}
                  \end{claim}
                  \begin{proof}
                    Choose a tuple $\bar c$ of atoms which contains more than twice the number of atoms needed to support any element of $\Sigma$ (and therefore also any element of the semigroup).
                    These atoms will be  the support of the function $E$.

                    The set  of $\Jj$-classes in the semigroup is itself an equivariant set -- as a quotient of an equivariant set under an equivariant equivalence relation~\cite[p.~59
                    ]{bojanczyk_slightly2018} --  and therefore one can talk about the support of a $\Jj$-class. 
                    By the same argument as in the proof of Lemma~\ref{lem:straight-choice}, if a $\Jj$-class contains an  idempotent, then it contains an idempotent $e$ which satisfies:
                    \begin{enumerate}
                    \item[($\dagger$)] the  least support of $e$ is contained in the union of:
                    \begin{itemize}
                        \item  the  tuple $\bar c$ fixed at the beginning of this proof;
                        \item  the least support of the $\Jj$-class of $e$.
                    \end{itemize}
                    \end{enumerate}
                    Call an idempotent \emph{special} if it satisfies property ($\dagger$) above, 
                    and the subset of  atoms from $\bar c$ in its least support is smallest with
                    respect to some fixed linear ordering of subsets of $\bar c$.

                    We define $E$ to be the uniformisation (Lemma \ref{lem:straight-choice}) of the following
                    binary relation on $\Sigma$:
                    \begin{align*}
                        \set{(a, e) : \text{$e$ is a letter representation of some special idempotent in the $\Jj$-class of $h(a)$}} 
                    \end{align*}
                    The only thing left to show is that  $E$ defined this way satisfies  property (b)
                    of the claim.
                    The function which maps a semigroup element to its $\Jj$-class is  equivariant, by the assumption that the semigroup is equivariant.
                    Equivariant functions can only make supports smaller, which means that the least support of every special idempotent $e$ contains the entire least support of the $\Jj$-class of $e$.  By definition, all special idempotents in the same $\Jj$-class use the same constants from $\bar c$ in their least support. Summing up, all special idempotents in the same $\Jj$-class have the same least support, namely the least support of the $\Jj$-class, plus some fixed atoms from $\bar c$ that depend only on the $\Jj$-class.
                    This observation extends to letter  representations, since they have the same least supports as the represented elements. 
                  \end{proof}
            \item In this step we substitute all $e_i$ with $e_1$. We do it in several substeps:
                  \begin{enumerate}
                    \item Equip each $e_i$ with $\supp{e_i}$ using Lemma \ref{lem:write-support}. (Homomorphism)
                    \item Propagate each $\supp{e_i}$ one position forward. (Delay function from the lemma below)
                    \begin{lemma}
                    \label{lem:delay}
                    For every polynomial orbit-finite $\Sigma$, the delay function 
                    \begin{align*}
                    a_1 \cdots a_k \in \Sigma^* \qquad \mapsto \qquad \vdash a_1 \cdots a_{k-1} \in (\Sigma+ \bot)^*
                    \end{align*}
                    is a composition of primes.
                    \end{lemma}
                    \begin{proof}
                    For every polynomial
                    orbit finite $\Sigma$, the {\em letter propagation}
                    function which works analogously to the atom
                    propagation prime function (Example \ref{ex:atom-propagation}), except that $\atoms$ is replaced by $\Sigma$, is a composition of primes.
                    The proof is a straightforward induction
                    on the construction of $\Sigma$.
                    
                    In order to define the delay function as a composition of primes, we do the following: 
                    Use a classical Mealy machine to mark each  position as even or odd. Next, using a homomorphism and letter propagation, 
                    propagate all  letters  on even positions to the next position (even positions send, odd positions receive). Next, do the same for letters on odd positions. 
                    \end{proof}
                    \item Let $d \in \set{0,1,\ldots}$ be the size of the least support of the idempotents $e_1,\ldots,e_k$. For positions $i < j$ define 
                    \begin{align*}
                        \tau^{i}_{j} : \set{1,\ldots,d} \to \set{1,\ldots,d}
                    \end{align*}
                    to be the permutation of coordinates which transforms the tuple $\supp{e_i}$
                    into the tuple $\supp{e_{j}}$. Such a permutation exists because all of $e_1,\ldots,e_k$ have the same least support. Using a homomorphism and the results of the delay function from the previous step, label each position $i$ with the permutation $\tau^{i-1}_i$.  (Homomorphism)
                    \item Compose the permutations from the previous step  to compute $\tau^1_i$ in each position. (Classical Mealy)
                    \item Compute $\supp{e_1}$ in every position by applying $\tau^1_i$ to $\supp{e_i}$. (Homomorphism)
                    \item Substitute every atom in $e_1$ with a placeholder according to its position in $\supp{e_1}$.
                          Call the result $e_1'$. Note that $e'_1$ is an atomless value. (Classical Mealy to mark the first position + Homomorphism + Definition~\ref{definition:placeholders})
                    \item Propagate $e_1'$ throughout the word. (Classical Mealy)
                    \item In every position, compute $e_1$ by substituting placeholders from $e_1'$ with
                          atoms from $\supp{e_1}$. (Homomorphism).
                 \end{enumerate}
                 From now on, let $e = e_1$.
            \item Let the list of pairs produced in the previous step be $(a_1, e), \ldots (a_n, e)$. In this step
                  we decompose each $h(a_i)$ into a product $h(a_i) = h(x_i)h(y_i)$ such that $h(e)$ is a suffix of $h(x_i)$
                  and a prefix of $h(y_i)$. The letter $e$ represents an idempotent, so this is equivalent
                  to $h(e, x_i) = h(x_i)$ and $h(y_i, e) = h(y_i)$. In order to compute the decomposition we
                  use the homomorphism prime function together with the following claim.
                  \begin{claim}
                    There is a finitely supported function $F : \Sigma^2 \to \Sigma^2$ that does this:
                    \begin{itemize}
                        \item \textbf{Input.} $(a, e)$, where $h(e)$ is an idempotent
                               from the $\Jj$-class of $h(a)$.
                        \item \textbf{Output.} A pair $(x, y)$, such that $h(a) = h(x) h(y)$, $h(x) = h(x) 
                        h(e)$, and $h(y) = h(e) h(y)$.
                    \end{itemize}
                  \end{claim}
                  \begin{proof}
                    If $h(e)$ is an idempotent from the $\Jj$ class of $h(s)$ then $h(e) = h(e) h(e)$ is an infix
                    of $h(s)$ and therefore there exists at least one factorisation $(h(x), h(y))$ as required by the
                    claim. To produce $F$ apply Lemma \ref{lem:straight-choice}.
                  \end{proof}
            \item For all positions $i \in \set{2, \ldots k}$ compute
            $g_i$ -- a letter representation of $h(y_{i-1}, x_{i})$. (Delay~+~Homomorphism)\\
            Here, we need to calculate a letter representation for a product of two elements. This is not immediately obvious -- because there might be a need to choose between several letter representations -- but the problem can be solved thanks to the following lemma. 
    
            \begin{lemma}\label{lem:straight-representation} There is an equivariant function $f : \Sigma \times \Sigma \to \Sigma$ such that for every $x,y \in \Sigma$, $f(x,y)$ is a letter representation of $h(x, y)$. 
            \end{lemma}
            \begin{proof} 
                Here it is important that the letter representation  has the same least support as the represented element (i.e.~the letter representation reflects supports).  Thanks to this property, we can apply~\cite[Claim 6.10]{bojanczyk_slightly2018}. 
            \end{proof}
            \item Let $g$ be a letter representation of
                  $h(g_2 \ldots g_k)$. 
                  Note that $h(a_1 \ldots a_k) =
                  h(x_1,g,y_k)$.  This means that in
                  order to compute a representation of the product of the entire block, we just need to collect all the values : $x_1$, $y_k$ and $g$ in the last position
                  and apply Lemma~\ref{lem:straight-representation}. The value $y_k$ is already there and we can propagate the value
                  $x_1$ to the last position using letter propagation from the proof of Lemma~\ref{lem:delay}.
                  This leaves us with computing $g$ in the last position. The crucial observation we need for that
                  is the following:
                  \begin{claim}
                    \label{claim:group-equal-supports}
                    All $g_i$ and $e$ have the same least supports.
                  \end{claim}
                  \begin{proof}
                    Define an $\Hh$-class to be any non-empty intersection of a prefix ($\Rr$-) class
                    and of a suffix ($\Ll$-) class. Each $h(g_i)$ begins with $h(e)$, ends in $h(e)$, and is in the
                    $\Jj$-class of $h(e)$. It follows that $h(g_i)$ is in the same $\Hh$-class as $h(e)$;
                    this is because $\Rr$- and $\Ll$- classes form an antichain in a given $\Jj$-class
                    \cite[Lemma 7.1]{bojanczykNominalMonoids2013}. Like for any idempotent, the
                    $\Hh$-class of $h(e)$ is a group \cite[Lemma 11]{colcombet2011green}. In an orbit-finite semigroup with an equivariant product operation,
                    all elements have equal least supports \cite[Lemma 2.14]{DBLP:journals/corr/ColcombetLP14}. This extends to their letter representations.
                  \end{proof}
                  Thanks to the above claim, we can compute $g$ by applying similar technique
                  as in step 2 -- replace every atom in each $g_i$ with a placeholder according to its
                  position in $\supp{e}$, obtaining an atomless value $g_i'$. Then, use a classical Mealy
                  machine to compute the atomless version of the product $g'$ in the last position (here we apply
                  Claim~\ref{claim:function-extend-to-placeholders} to the binary version of $h$,
                  so that we can use it on values containing placeholders).
                  Finally,
                  we replace all the placeholders in $g'$ with atoms from $\supp{e}$, computing the value $g$.
        \end{enumerate}
    \end{proof}
    Now, we are ready to finish the induction step
    for the Split Lemma. We
    start by partitioning the input into blocks of the following form:
    \begin{align}\label{eq:blocks}
        \underbrace { \overbrace{a_1 \cdots a_k}^{\txt{\scriptsize non-empty \\ \scriptsize smooth sequence}} {a_{k+1}}}_{\textrm{not a smooth sequence}}
    \end{align}
    The last block of the word might be unfinished -- it might not have the final element that breaks the smoothness. The partition is represented represented by distinguishing the last position ($a_{k+1}$) in each block with a symbol $\dashv$.
    We construct the partition in the following steps, using locality of the smooth product
    (Lemma \ref{lem:smooth-local}):
    \begin{enumerate}
        \item Apply the delay function (Lemma \ref{lem:delay})
        \item Mark all positions $i$, such that $a_{i-1} a_i$ is not smooth. (Homomorphism)
        \item By  locality of smooth sequences, every factor without marked positions (as in the previous item) is a smooth sequence. A small problem is that there might be consecutive marked positions, which would lead to empty smooth sequences between marked positions. To solve this problem, for every block of consecutive marked positions, use the marker $\dashv$ only for every second one  (the first, the third, and so on).
              (Classical Mealy).
         \end{enumerate}
    From now on, we use the name \emph{distinguished position} for the positions marked by $\dashv$ above. Define an \emph{inductive block} to be a block of positions where all positions (except possibly the last one) are not distinguished, and which is maximal for this property. The string positions are partitioned into inductive blocks, and each inductive block except possibly the last one  has shape as in~\eqref{eq:blocks}.

    We now  compute a representation of the product for every inductive  block and write it down in its last position -- if the last inductive block is unfinished (i.e.~it does not end with a distinguished position), this operation has no effect.
    By applying the $\mathsf{map'}$ combinator from Lemma \ref{lem:prime-combinators} we only need to show how to do it for
    one complete block $a_1, \ldots, a_k, a_{k+1}$:
    \begin{enumerate}
        \item Apply the delay function. (Lemma \ref{lem:delay})
        \item Now the $a_k$ is in the ($k+1$)-th position and ``sees'' the $\dashv$ endmarker. Thanks to that, we can apply
              Lemma \ref{lem:smooth-product-prime} to compute a letter representation of
              $h(a_1 \cdots a_k)$
              in the last position (call it $p$). (Lemma \ref{lem:smooth-product-prime})
        \item In the last position (marked with $\dashv$) calculate
              a letter representation of
              $h(p, a_{k+1})$
              to obtain a letter representation of the product of the entire block. (Homomorphism + Lemma \ref{lem:straight-representation})
    \end{enumerate}

    Now, we have a situation like this:
    \[
        \begin{matrix}
            a_1 & a_2 & \ldots & a_{k_1}    & a_{k_1+1} & \ldots     & a_{k_2}       &  \ldots & & a_{k_m} & a_{k_m + 1}& \ldots\\
                &     &        & \dashv &         &            & \dashv        &  \ldots & &\dashv   &            &       \\
                &     &        & a'_1   &         &            & a'_2          &  \ldots & &a'_m     &            &       \\
        \end{matrix}
    \]
    Note that all $h(a'_i)$ belong to a subsemigroup  that has a smaller $\Jj$-height.
    This is because all $h(a'_i)$ have proper infixes
    by construction.
    Apply the induction assumption with
    the subsequence combinator (Lemma \ref{lem:prime-combinators}) to compute a smooth split for the string restricted to distinguished positions. Let $n$ be the height of the split from the induction assumption. For each distinguished position, increment its height by 1. For each non-distinguished position $i$, define its height to be 1, and its split value to be $a_i$. 
    
\section{Mealy machines as compositions of primes}
\label{ap:mealy-primes}
In this section of the appendix, we prove Theorem \ref{thm:mealy-machine-krohn-rhodes}, which says that every total function
defined by a single-use Mealy machine is a composition  of primes.

Before proceeding with the proof, we argue why the converse of Theorem~\ref{thm:mealy-machine-krohn-rhodes} is also true, i.e.~every composition of primes is a single-use Mealy machine. This is because  all the prime functions are clearly single-use Mealy machines, and single-use Mealy machines are closed under the two kinds of composition. There is a slightly subtle point in preserving the single-use restriction for  the sequential composition, so we present this proof in more detail:

\begin{lemma}\label{lem:composition}
    Both single-use one-way transducers and single-use Mealy machines are closed under sequential composition.
    \end{lemma}
    \begin{proof}
    The only problem with the classical product construction for $\Aa \circ \mathcal{B}$ is
    that $\mathcal{A}$ might ask for multiple copies of $\mathcal{B}$'s output,
    whereas $\mathcal{B}$ emits every output letter only once. In order to solve this problem,
    we use the following claim:
    \begin{claim}\label{claim:stay-in-place-bound}
        For every  single-use Mealy machine,
        there is a bound $k \in \set{0,1,\ldots}$, such that if the machine stays in one position for more than $k$ steps,
        then it will loop and stay there forever. The same is true for single-use one way transducers.
    \end{claim}
    \begin{proof}
        As long as a single-use register machine stays in one place, each of its register may either: (a) store the atom that was present in the register when the machine entered its current position; (b) be undefined; (c) store one of the atoms taken from the input letter under the head.
        It follows that there are at most 
        \begin{align*}
            \text{(number of states)} \cdot (2 + \text{(maximal number of atoms in an input letter)})^{\text{number of registers}}
        \end{align*}
         possible values for the state and register
        valuation. If the machine stays
        in one place for more than that, it will visit some state and register
        valuation for the second time and start to loop. 
    \end{proof}
    This means that if the product construction keeps $k$ copies of $\mathcal{B}$, it will never
    run out of values to feed to $\mathcal{A}$. 
    \end{proof}

The rest of Appendix~\ref{ap:mealy-primes} is devoted to proving Theorem~\ref{thm:mealy-machine-krohn-rhodes}. 
Fix a single-use Mealy machine  that defines a total function
$f: \Sigma^* \to \Gamma^*$. Both alphabets $\Sigma$ and $\Gamma$ are polynomial orbit finite sets.
To show that $f$ is a product of
primes, we will (a) apply the Split Lemma from~\ref{sec:factorisation-forests} to the input string for a suitably
defined homomorphism; and then (b) use the smooth split to produce the output string.
One advantage of this strategy is that a similar one will also work for two-way single-use transducers,
as we will see in Section~\ref{ap:two-way-to-composition}.

\subsection{State transformation monoid}
\label{sec:Mealy-transformation-monoid}
In this section we describe the monoid homomorphism that will be used in the smooth split.
Roughly speaking, this homomorphism is similar to the one of Shepherdson functions discussed in Section~\ref{sec:automata-to-semigroups}, restricted for the one-way model of single-use Mealy machines. As was the case for Shepherdson functions, orbit-finiteness
of the resulting monoid crucially depends on the single-use restriction.

\begin{definition}
    Define the {\em extended states} of the fixed single-use Mealy machine to be 
    \begin{align*}
    \bar Q  \eqdef Q \times \text{(register valuations)} + \lightning .
    \end{align*}
    The element $\lightning$ represents computational error, 
    resulting from accessing undefined registers.
\end{definition}
The set  $\bar Q$ of extended states is a polynomial orbit-finite set.
There is a natural right action of input strings on extended configurations: for an extended configuration $\bar q \in \bar Q$ and an input string $w \in \Sigma^*$, we write $\bar qw \in \bar Q$ for the extended configuration after reading $w$ when starting in $\bar q$. If $\bar q$ is $\lightning$, then $\bar q w$ is also $\lightning$. Define 
\begin{align*}
h: \Sigma^* \to M
\end{align*}
to be the monoid homomorphism that maps a word $w$ to the transformation $\bar q \mapsto \bar q w$. 
By the same reasoning as for Shepherdson functions
(Section \ref{sec:automata-to-semigroups}), $M$ is an orbit-finite monoid, which we call the \emph{state transformation monoid}. 
By definition of $h$, there is also a right action of $M$ on $\bar Q$ (function application) which  commutes with $h$:
\begin{align*}
\bar q \, h(w) \stackrel{\textrm{def}}{=}  h(w)(q) = qw.
\end{align*}
We say that an extended state $\bar q$ is \emph{compatible} with  a monoid element $m$ if $\bar q m \neq \lightning$. 

We end this subsection with some results about  Green's relations for the state transformation monoid.

\begin{lemma}
    \label{lem:prefix-equicompatible}
    If $x \in M$ is a prefix of $y \in M$, then every $\bar q \in \bar Q$ that is compatible with $y$
    is also compatible with $x$. 
\end{lemma}
\begin{proof}
    There exists an $x' \in M$, such that $y = xx'$. Therefore, if $x$ is not compatible with $\bar q$,
    then 
    \[\bar q y = \bar q (x x') = (\bar q x) x' = \lightning x' = \lightning \]
    which means that $y$ is also not compatible with $\bar q$.
\end{proof}
The following is an important corollary of Lemma \ref{lem:prefix-equicompatible} and Green's Eggbox Lemma
(Lemma~\ref{lem:green-eggbox}).
\begin{corollary}
    \label{cor:smooth-same-registers}
    For every smooth sequence $x_1,  \ldots, x_k$ in the state transformation monoid, an  extended state
    is compatible with $x_1$ if and only if it is compatible with $x_1 \cdots x_k$. 
\end{corollary}

\subsection{Masked states}
In order to prove Theorem~\ref{thm:mealy-machine-krohn-rhodes} it would be useful to  compute the run of the single-use Mealy machine on the input word, i.e.~decorate  every input position with  the extended state of the machine after reading the input string up to (but not including) that position.
Unfortunately, for similar reasons as the ones pointed out in Example \ref{ex:no-compute-state} this is not
always possible -- outputing the content of every register in ever position
can violate the single-use restriction.
To overcome this problem, we will ``mask'' some registers in the run, so that there is no need to output them multiple times. 
\begin{definition}[Masked states]
    \label{def:mased-states}
    An extended state $\bar q_1$ is
    said to be
    a {\em masking} of
    another extended state $\bar q_2$ if $\bar q_1$ can be obtained from $\bar q_2$ by setting a (possibly empty) subset of the  registers to the undefined value ($\bot$).
    In particular $\lightning$ is the only masking of $\lightning$.
    For an extended state $\bar q$ we write $\bar q \downarrow$ for the set of all extended states that can be obtained by  masking $\bar q$.
    Observe that $(\bar q m) \downarrow = (\bar q \downarrow) m - \set{\lightning}$, for every $m \in M$ which is compatible with $\bar q$.
\end{definition}

Now we are ready to state the main lemma of this section:
\begin{lemma}
    \label{lem:reconstuct-run}
    The following function is a composition of primes:
    \begin{itemize}
        \item \textbf{Input:} A word over $\Sigma \times \bar Q + \Sigma$, such that only the first letter contains an extended configuration:
        \begin{align*}
            \begin{matrix}
            a_1& a_2& \ldots & a_n \\
            \bar q_1&    &        &
            \end{matrix}
        \end{align*}
        \item \textbf{Output:} A sequence of extended configurations
       \begin{align*}
        \begin{matrix}
        \bar q_1& \bar q_2 & \ldots & \bar q_{n}
        \end{matrix} 
    \end{align*}
such that for every $i \in \set{1,\ldots,n}$, $\bar q_i$ is a masking of the extended configuration
\begin{align*}
\bar p_i =  \bar q_1 a_1 \cdots a_{i-1}
\end{align*}
and furthermore, if $\bar p_i$ is compatible
with $a_i$, then so is $\bar q_i$.
\end{itemize}
\end{lemma}
We now argue how  Theorem \ref{thm:mealy-machine-krohn-rhodes} is a consequence of the lemma above. In order to compute machine's output
for an input word $a_1 \ldots a_n \in \Sigma^*$, it suffices to (a) 
set $\bar q_1$
to be the initial configuration of the Mealy machine, (b) apply Lemma \ref{lem:reconstuct-run}, and (c)
apply a homomorphism that simulates the run of each $\bar q_i$ on $a_i$ and returns
the output letter. Note that, thanks to the compatibility condition in Lemma~\ref{lem:reconstuct-run}, in order to perform step (c)
it is enough to have the masking $\bar q_i$ (and not the real extended state $\bar p_i$).
The rest of this section is dedicated 
to proving Lemma~\ref{lem:reconstuct-run}.

\subsubsection{Proof of Lemma \ref{lem:reconstuct-run}}
By Lemma~\ref{lem:extend-to-star}, we can assume without loss of generality that the homomorphism $h : \Sigma^* \to M$ satisfies  assumption (*) in the Split Lemma.  Apply the \SplitLemma to $h$. 
This gives us a function (that is a composition of primes) which maps each input string  to a smooth split of bounded height, with respect to the state transformation monoid. Whenever we talk about \emph{the smooth split} of some input string in $\Sigma^*$, we mean the split produced by this function. The proof of Lemma~\ref{lem:reconstuct-run} is by induction on  the height of the smooth split for the input string (we call this the \emph{split height} of the input string).

\smallparagraph{Split height 1}
If the split height of the input string is equal to $1$, then the entire input is a smooth sequence and we can use the following claim (which we will also use 
in the induction step):
\begin{claim}
    \label{claim:reconstruct-run-aux-smooth}
    Lemma \ref{lem:reconstuct-run} is true with extra assumption that $a_1 \cdots a_n$ is smooth, i.e.~the sequence~$h(a_1),\ldots,h(a_n)$ is smooth. 
\end{claim}
\begin{proof}
    In this proof, we  use a convention
that letter representations   are written in \blue{blue}.
Changing the colour of a variable from blue
to black, denotes going from a letter representation to the represented element, i.e.~applying the homomorphism $h$. Lemma \ref{lem:straight-representation}
states that the binary product from $M$ can be lifted to
work on letter representations. This lifted product might, however,
not be associative : $(\blue{x} \blue{y}) \blue{z}$ and
$\blue{x} (\blue{y} \blue{z})$ might be different representations
of the same element  $xyz$.

    Start the construction by repeating the first four steps of the construction in the proof of Lemma \ref{lem:smooth-product-prime}.
    As a result, in every position $i$ we have
    (a) ${a_i}$; 
    (b) a letter representation $\blue e \in \Sigma$ of an idempotent $e$
        from the same $\Jj$-class as  $h({a_i})$  -- in every position
        the same representation of the same
        idempotent;
    (c) letter representations $\blue {x_i},\blue {y_i} \in \Sigma$ such
        that $h(a_i) = x_i y_i$ and $e$ is a
        suffix of $x_i$ and a prefix of $y_i$; and
    (d)  a letter representation $\blue {g_i} \in \Sigma$ such that $g_i=  y_{i-1} x_{i}$ for $i > 1$.
    To obtain the intermediate extended states, we proceed with the following steps:
    \begin{enumerate}
        \item In the first position compute $\bar q_1 x_1$. (Homomorphism)
        \item Mask out (by replacing their values with $\bot$) all the registers in $\bar q x_1$ that contain atoms which do not appear in the least support  $\supp{\blue{e}}$. Call the result $\bar p_1 \in \bar q_1 x_1 \downarrow$. (Homomorphism)
        \item Propagate $\bar p_1$ throughout the word:
              \begin{enumerate}
                \item Replace every atom in $\bar p_1$ with a placeholder according
                to its position in $\supp{\blue e}$. Call the result of this replacement $p_1'$. Note that this element contains no atoms.  (Homomorphism +
                Definition \ref{definition:placeholders} + Lemma~\ref{lem:write-support})
                \item  Propagate $p_1'$ to every  position. (Classical Mealy)
                \item In every position,  replace the  placeholders in $p_1'$ with the original atoms from $\supp{\blue e}$,  thus recreating $\bar p_1$. (Homomorphism).
              \end{enumerate}
        \item In every position compute $\blue{\overrightarrow{g_i}}$ - a letter representation of $g_2 \ldots g_i$ with the corner case of $\blue{\overrightarrow{g_1}}=\blue{e}$:
        \begin{enumerate}
            \item Replace every atom in every $\blue {g_i}$ with a placeholder according to its position in $\supp{\blue e}$.
                  Call the result $\blue {g'_i}$. Thanks to Claim \ref{claim:group-equal-supports} we know that $g_i'$
                  contains no atoms. (Homomorphism +
                  Definition \ref{definition:placeholders} + Lemma~\ref{lem:write-support})
            \item Thanks to Claim~\ref{claim:function-extend-to-placeholders}
                  we can extend the binary product on letter representations  to work with placeholder-values $\blue{g'_i}$. They are all atomless,
                  so their prefix products $\blue{\overrightarrow{g_i}'}$
                  can be computed by a classical Mealy machine. 
            \item Replace the placeholders in each $\blue {\overrightarrow{g_i}'}$
                  back with atoms from $\supp{\blue e}$, obtaining $\blue {\overrightarrow{g_i}}$. (Homomorphism + Claim~\ref{claim:function-extend-to-placeholders})
        \end{enumerate}
        \item In every position $i$, compute $\bar q_{i+1} = \bar p_1 \overrightarrow{g_i} y_i$. (Homomorphism)
        \item Send all $\bar q_{i+1}$ one position forward. (Delay from Lemma \ref{lem:delay})
    \end{enumerate}
    To prove the correctness the construction,
    we first notice that if $\bar q_1 x_1 = \lightning$, then for every $i > 1$, the construction computes $\bar q_i = \lightning$
    which is the correct answer. From now on we assume that $\bar q_1$ and $x_1$ are compatible.
    Notice that $\bar q_1 x_1 e = \bar q_1 x_1$. This means that
     $e$ cannot
    use any register from $q_1 x_1$ which has an atom that is not
    present in $e$, because $e$ would not have been able to restore it.
    This means that $p_1 \in q_1 x_1 \downarrow$ is compatible with $e$.
    From Corollary~\ref{cor:smooth-same-registers} we obtain that
    it is also compatible with $y_2 x_3 \ldots y_i$ for every $i$.
    This is because $y_2 x_3 \cdots y_k$ is equal to  $e y_2 x_3 \cdots y_k$ and $e, y_2, x_3, \ldots, y_k$ is a smooth sequence. 
\end{proof}
\smallparagraph{Induction step}
After constructing the smooth split on the input we get the following:
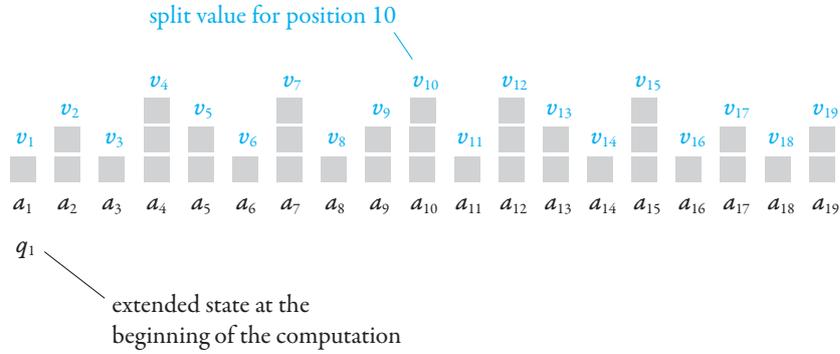
\mypic{28}
We start the construction for the induction step, 
 by propagating $q_1$ to the first position with the maximal height (in the picture above it is position $4$).
Notice that the subsequence of split-values $v_i$ for $i$ of maximal height is smooth.
Making use of that, we apply Claim \ref{claim:reconstruct-run-aux-smooth} to this subsequence
with the initial state being $\bar q_1$.
Now, in every position $i$ of maximal height we have computed the value $\bar q_{j+1}$, where $j$ is the index of the previous position of maximal height -- in the first position of maximal height we
have $j = 0$. This is illustrated in the following picture:
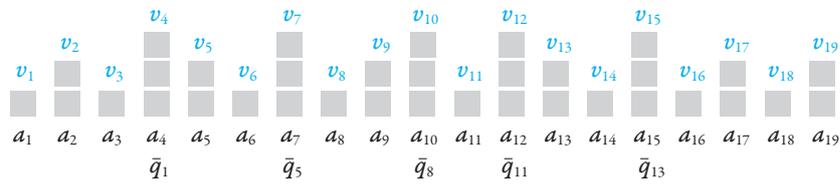
\mypic{30}
We proceed by computing (in each $i$ of maximal height) the
values $\bar q_{i+1} = \bar q_{j+1} v_i$ and $\bar q_{i} = \bar q_{j+1} u_i$ (where $u$ is the strict descendant split value from Remark \ref{rk:split-exclusive-values}). After that, we propagate the $\bar q_{i+1}$ values one position to the right (Lemma \ref{lem:delay}).
\mypic{29}
We divide the input into blocks, treating the positions with maximal height as separators. Every such
block has a lower split height and has its initial state written down, so we can apply the induction
assumption to every block (using Lemma~\ref{lem:prime-combinators}), obtaining a $\bar q$-value in every position, finishing the construction\\
Now, we proceed with the proof of correctness, which almost entirely follows immediately from
the construction. The only thing that is unclear is whether the first states of the blocks
have not been masked too much. More formally it boils down to proving the following:
\begin{claim}
    Define $\bar p_j = \bar q_1 m_1 \ldots m_{j-1}$ (like in the statement of Lemma \ref{lem:reconstuct-run}).
    Let $i$ be the first position of a block and let $k$ be the size of this block, then for every
    $j \in \set{i, \ldots, k-1}$ if $\bar p_i$ is compatible with $m_i \cdots m_j$, then so is $\bar q_i$.
\end{claim}
First, we show that this is true for every block, but the last one. Then we present a construction that fixes any potential errors in the last block.
\begin{proof}
    For the first block the Claim is immediate, since $q_1 = \bar p_1$. Let $t$ be the first position of maximal height.
    If $q_1$ and $v_t$ are not compatible, then for all $j > t$, we have $\bar p_i = \bar q_i = \lightning$,
    which also makes the Claim immediate. From now on assume that $q_1$ and $v_t$ are compatible. Take
    $i$ and $k$ as in the statement of the Claim with the extra assumption that $i$ is not in the last block. 
    This means that $i + k$ is the first position of the next block. From Claim \ref{claim:reconstruct-run-aux-smooth} and
    \ref{cor:smooth-same-registers} we obtain that $\bar q_i$ is compatible with $v_{i+k}$. For every $j < k$,
     $m_i \ldots m_j$ is a prefix of $v_{i + k}$. The Claim follows
    from  Lemma \ref{lem:prefix-equicompatible}.
\end{proof}
The easiest way to deal with the last block would be to send $q_1$, together with
all the most recent register updates to the first position of the last block and ``unmask'' all the registers
in its extended state. This is impossible, because Mealy machines (and compositions of primes) cannot detect the last block -- they are one-way models, that are additionally unaware of the end of input.
Instead, we transfer all the necessary information to the first position   where $\bar q_i$ and $m_i$ are incompatible:
\begin{enumerate}
    \item Mark the first position $i$ such that $q_i$ and $m_i$ are incompatible (if it exists). (Homomorphism + Classical Mealy)
    \item Send $q_0$ to $i$. (Letter propagation)
    \item In every position smaller than $i$, compute which registers have been changed in this position
          and to which value (this is possible, because in positions
          smaller than $i$, values $\bar q$ and $m$ are compatible). Propagate the most recent values of every register to $i$.
          (Classical Mealy + Atom propagation)
    \item Using $q_0$ and the information about modifications of each register, restore values of all the masked registers in 
          $q_i$. Since we are only interested in the
          case where $i$ is in the last block, we can assume that the
          suffix starting in $i$ has a lower split-height. This means
          we can apply the induction assumption to this suffix, with
          initial configuration equal to the restored version of $\bar q_i$.
\end{enumerate}
Since the only actual error can appear in the last blocks, those four error-fixing states only need to be applied once.

\section{Proof of Theorem \ref{thm:one-way-two-way}}
\label{sec:moniod-to-automaton}
In this short section of the appendix, we finish the proof of
Theorem \ref{thm:one-way-two-way}, by showing how to construct
a one-way single-use automaton for a language recognised by a homomorphism
to an orbit-finite monoid. We will do it by composing (a) a function
that appends $\dashv$ to the input word, with (b) a function
that computes a smooth split for a suitable monoid morphism, and with
(c) a simple single-use one way automaton.
This will finish the proof, because
it follows from  \SplitLemma and Theorem \ref{thm:mealy-machine-krohn-rhodes}
that there is a single-use transducer that computes the smooth split and
it follows from Lemma \ref{lem:composition} that a composition of
a single-use one-way transducers with a single-use one-way automaton
can be expressed a single-use one-way automaton.  

Take any language $L$ over a polynomial orbit-finite alphabet $\Sigma$, 
recognised by an equivariant morphism $h : \Sigma^* \to M$ and an equivariant $F \subseteq M$. Define $M^0$ to be $M$ equipped with a $0$ element, such that
$m 0 = 0 m = 0$, for every $m \in M$.
Define $h' : (\Sigma + \dashv)^* \to M^0$, such that $h'(\dashv) = 0$, and
$h'(m) = h(m)$, for $m \in M$. Notice that every element is an infix of $0$,
but $0$ is not an infix of any other element. Because of that, for every
$w \in \Sigma^*$, in the smooth split of the word $w\dashv$ for $h'$
only the last position has the maximal height. This means that
the auxiliary strict-descendant value from Remark~\ref
{rk:split-exclusive-values} in the last position is equal to $h(w)$.
The characteristic function of $F$ is equivariant (because $F$ is
equivariant), so a single-use one-way automaton can easily check whether
$h(w) \in A$.
\section{Equivalence of the transducer models}
\label{sec:appendix-transducers}
In this part of the appendix, we prove equivalence for all of the transducer models in  Theorem~\ref{thm:two-way-models}.  The proof is spread across six subsections, and its plan is illustrated in Figure~\ref{fig:proof-plan}. Note that from the fact that two-way single-use transducers are equivalent to compositions
of primes (Sections \ref{ap:composition-to-two-way} and \ref{ap:two-way-to-composition}) it follows that
two-way single-use transducers are closed under
compositions. We will rely on this fact
when translating regular list functions and
streaming string transducers into two-way-transducers  (Sections \ref{ap:regular-list-functions-to-two-way} and \ref{ap:sst-to-two-way}).

\begin{figure}%
    \centering
    \begin{align*}
        \xymatrix@C=3cm@R=2cm{
            \small
            &
            \txt{two-way \\transducers}
            \ar@/^1pc/[dd]^{\txtlabel{Section~\ref{ap:two-way-to-composition}}}
            &\\
            \txt{regular list\\ functions}
            \ar[ur]^{\txtlabel{Section~\ref{ap:regular-list-functions-to-two-way}}}
            &
            &
            \txt{streaming string\\ transducers}
            \ar[ul]_{\txtlabel{Section~\ref{ap:sst-to-two-way}}}
            \\
            &
            \txt{composition of\\ two-way primes}
            \ar[ul]^{\txtlabel{Section~\ref{ap:composition-to-regular-list-functions}}}
            \ar@/^1pc/[uu]^{\txtlabel{Section~\ref{ap:composition-to-two-way}}}
            \ar[ur]_{\txtlabel{Section~\ref{ap:composition-to-sst}}}
        }
        \end{align*}
    \caption{Proof plan for Theorem~\ref{thm:two-way-models}}
    \label{fig:proof-plan}
\end{figure}

Before proceeding with the proof, we illustrate the importance of the single-use restriction by showing that equivalence fails when the single-use restriction  is lifted. The single-use restriction appears in two of the models from Theorem~\ref{thm:two-way-models}, namely single-use two way transducers and streaming string transducers with atoms. When talking about multiple-use  streaming string transducers with atoms, we lift the  single-use restriction from both string registers and atom registers. One could imagine intermediate models, where the single-use restriction is used only for atom registers but not for string registers, or the other way round; these models would also be non-equivalent.

\begin{theorem}\label{thm:non-equivalent-copyful} None of the following models are equivalent: 
    \begin{enumerate}
        \item any of  the equivalent single-use models  from  Theorem~\ref{thm:two-way-models};
        \item multiple-use two-way transducers;
        \item multiple-use streaming string transducers with atoms. 
    \end{enumerate}
\end{theorem}
\begin{proof}(sketch)
    The model in item 1 is clearly included in the models from items 2 and 3. We now show that this inclusion is strict, and the models from items 2 and 3 are incomparable.
    
    Multiple-use  streaming string transducers atoms can have exponential size outputs, by  duplicating a string register in each step. This is in contrast with the  model from item 1, which has linear size increase (Lemma~\ref{lem:bounded-width}), and with the model from item 2, which has  polynomial size increase due to the number of configurations that can be used in a non-looping run.  Therefore, the model in item 3 is not contained in any of the others.
    
    To see why item 2 is not contained in the others, consider the string-to-boolean function
    \begin{equation}\label{eq:all-distincts}
    w \in \atoms^* \quad \mapsto \quad 
    \begin{cases}
     \text{yes} & \text{if all letters  are pairwise distinct}   \\
     \text{no} & \text{otherwise, i.e.~some letter appears twice.}
    \end{cases}
    \end{equation}
    This function is computed by a  multiple-use two-way transducer, see~\cite[Example 11]{kaminskiFiniteMemoryAutomata1994}. 
    To see why the function is not computed by the other models, we use the following closure property that holds for any function from item 3, but which does not hold for the function from~\eqref{eq:all-distincts}.  Consider a nondeterministic orbit-finite automaton ({\sc nofa}), as defined in~\cite[p. 85]{bojanczyk_slightly2018}. Using a natural construction, one can show that if   $f : \Sigma^* \to \Gamma^*$ is as in item 3,  then languages recognised by {\sc nofa}  are preserved under inverse images of $f$.  On the other hand, languages recognised by {\sc nofa} are not closed under inverse images of the function from item~\eqref{eq:all-distincts}, since otherwise the inverse image of the language $\set{\text{yes}}$, i.e.~the words with all letters pairwise distinct, would be recognised by a {\sc nofa},  which it is not~\cite[Proposition 5]{kaminskiFiniteMemoryAutomata1994}.
    \end{proof}

The non-equivalence result in Theorem~\ref{thm:non-equivalent-copyful} is typical of the non-robustness of automata models for infinite alphabets. It is therefore all the more remarkable that, thanks to the single-use restriction, one can prove nontrivial equivalences such Theorem~\ref{thm:two-way-models}.

The rest of Appendix~\ref{sec:appendix-transducers} is devoted to the proof of Theorem~\ref{thm:two-way-models}, following the plan illustrated in Figure~\ref{fig:proof-plan}.

\subsection{Compositions of two-way primes to two-way transducers}
\label{ap:composition-to-two-way}
We begin by showing the implication \ref{two-way:krohn-rhodes} $\Rightarrow$ \ref{two-way:two-way} in Theorem~\ref{thm:two-way-models}, which says that 
every composition of prime functions can be computed by a single-use two-way transducer. There are three kinds of prime functions: homomorphisms, single-use Mealy machines, map reverse and map duplicate. In order to prove the implication, we use the classical Krohn-Rhodes Theorem to further simplify these prime functions. These simplified functions -- which we call the \emph{two-way primes} -- will also be used later in the proof.

\subsubsection{Two-way primes}
\label{sec:two-way-primes}
 In the classical Krohn-Rhodes theorem, there are two kinds of prime functions:
\begin{enumerate}
    \item \emph{Group transducers.} For a finite group $G$, consider  the length-preserving function of type $G^* \to G^*$, where the $i$-th output letter is the product (in the group), of the first $i$ input letters. Here is an example of the group transducer for the  group $\set{0,1,2}$ equipped with addition modulo 3:
    \begin{align*}
        \begin{array}{rccccccccccccccccccc}
            \text{input} & 1&2&0 & 0 & 2 & 1 & 0 &1  & 1 & 2 & 2 \\
            \text{output} & 1 & 0 & 0 & 0 & 2 & 0 & 0 & 1 & 2 & 1 & 0
        \end{array}
        \end{align*}
    \item \emph{Flip-flop transducer.} This is the function  with input alphabet $\set{a,b,1}$ and output alphabet $\set{a,b}$ that is recognised by the (classical) Mealy machine in the following picture, with each transition labelled by  (input letter / \red{output letter}):
    \mypic{15}
    The general idea is that the $i$-ith output letter is labelled by the most recent label in positions $<i$ that is other than $1$; if no such label exists then label $a$ is used. 
    Here is an example of the flip-flop transducer:
    \begin{align*}
        \begin{array}{rccccccccccccccccccc}
            \text{input} & 1&1 &b & 1 & 1 & b & 1 &1  & a & b & b \\
            \text{output} & \red a & \red a & \red a & \red b & \red b & \red b & \red b & \red b & \red b & \red a & \red b
        \end{array}
        \end{align*}
\end{enumerate}
The classical Krohn-Rhodes theorem says that every classical Mealy machine can be decomposed, using sequential and parallel composition, into functions which are either length-preserving homomorphisms (over finite alphabets), group transducers (for finite groups), or the flip-flop transducer.
In the following lemma, we use the classical Krohn-Rhodes Theorem to  simplify the prime functions mentioned in item~\ref{two-way:krohn-rhodes} of Theorem~\ref{thm:two-way-models}.

\begin{lemma} \label{lem:two-way-primes}
    Every function from item~\ref{two-way:krohn-rhodes} in Theorem~\ref{thm:two-way-models} is a sequential composition\footnote{We do not use parallel composition, which does not make sense for functions that are not length-preserving, which is the case for item~\ref{two-way-prime:homo} and~\ref{two-way-prime:map} in the lemma.  } of the following kinds of functions:
    \begin{enumerate}
        \item \label{two-way-prime:homo} an  equivariant homomorphism $\Sigma^* \to \Gamma^*$, where $\Sigma$ and $\Gamma$ are polynomial orbit-finite; or
        \item \label{two-way-prime:map}  map reverse or map duplicate; or
        \item \label{two-way-prime:kr} a parallel product $f|id$ where $id : \Sigma^* \to \Sigma^*$ is the identity function for some polynomial orbit-finite set, and $f$ is:
        \begin{enumerate}
            \item atom propagation; or
            \item a group transducer; or
            \item the flip-flop transducer.
        \end{enumerate}
    \end{enumerate}
\end{lemma}     
\begin{proof}
    The class of functions from item~\ref{two-way:krohn-rhodes} in Theorem~\ref{thm:two-way-models} are  the same as in the statement of the lemma, except that instead of item~\ref{two-way-prime:kr} of the lemma, Theorem~\ref{thm:two-way-models} uses single-use Mealy machines. By Theorem~\ref{thm:mealy-machine-krohn-rhodes}, every single-use Mealy machine can be decomposed -- using parallel and sequential composition -- into length-preserving homomorphisms, atom propagation, and classical Mealy machines. By the classical Krohn-Rhodes Theorem, classical Mealy machines can be decomposed -- using parallel and sequential composition -- into length-preserving homomorphisms, group transducers and the flip-flop transducer. Summing up, every single-use Mealy machine can be decomposed -- using parallel and sequential composition  -- into length-preserving homomorphisms, atom propagation, group transducers and the flip-flop transducer. Finally, we  push the  parallel composition into the prime functions  by using the rules
    \begin{align}\label{eq:push-parallel-into-primes}
    (f_1 | f_2) =  (f_1|id) \circ (id|f_2) \quad (f_1 \circ f_2)|id = (f_1|id) \circ (f_2 |id) .
    \end{align}
\end{proof}  

Since we will frequently use the functions from the above lemma, we give them a name. 
\begin{definition}
    [Two-way primes] 
    \label{def:two-way-primes}
    Define a {two-way prime} to be any function as in items~\ref{two-way-prime:homo} -- \ref{two-way-prime:kr} of Lemma~\ref{lem:two-way-primes}. Define a \emph{composition of two-way primes} to be any sequential  (not parallel) composition of two-way primes. For the latter class of functions, we also  use the notation
    \begin{align*}
    \text{(two-way primes)}^*.
    \end{align*}
    
\end{definition}

\subsubsection{Pre-composition with two-way primes}
Thanks to Lemma~\ref{lem:two-way-primes}, the implication \ref{two-way:krohn-rhodes} $\Rightarrow$ \ref{two-way:two-way} in Theorem~\ref{thm:two-way-models} can be stated as the following inclusion:
\begin{align*}
\twoinclusion{(two-way primes)*}{single-use two-way}
\end{align*}
To prove the above inclusion, it is enough to show that  single-use two-way transducers are closed under pre-composition with   two-way primes:
\begin{align*}
\threeinclusion{(single-use two-way)}{(two-way primes)}{single-use two-way}.
\end{align*}
This is shown in the following lemma.

\begin{lemma}
    If $f$ is a single-use two-way transducer, and $g$ is a two-way prime function, then $f \circ g$ is a single-use two-way transducer.
\end{lemma}
\begin{proof}    
Suppose that $g$ is 
\begin{align*}
\text{map reverse} : (\atoms+|)^* \to (\atoms+|)^*.
\end{align*}
The transducer for $f \circ g$, is the same as $f$ (in particular, it has the same input and output alphabets), with the following differences. When the head of  $f \circ g$ is over an endmarker, it behaves the same way as $f$. When the head of  $f \circ g$ is over a  letter from $\atoms$, it behaves the same way as $f$, except that the ``previous'' and ``next'' actions are swapped. When the head of $f \circ g$ is over  a  separator $|$ which separates two blocks of atoms, it behaves the same way as $f$, with the following change (which the reader can easily extend to the corner cases where the previous and next blocks of atoms are empty):
    \mypic{16}    
    
    Similar straightforward constructions work when $g$ is any of the other prime two-way functions, and are left to the reader. In fact, there is a general explanation of why these constructions are simple: namely the two-way primes are recognised by two-way single-use transducers which are reversible (i.e.~every node in the graph of configurations has indegree and outdegree at most one), see~\cite[p.~2]{DartoisFJL17}. As observed in~\cite[Theorem 1]{DartoisFJL17}, if $f,g$ are two-way transducers such that $g$ is reversible, then a natural product construction gives a two-way transducer for the composition $f \circ g$; this construction can be applied also in the context of this lemma.  
\end{proof}

We would like to remark that the approach taken in this section could be of interest already  for the classical models for  finite alphabets. One of the important transducer constructions is the tree-trimming construction of Hopcroft and Ullman~\cite[p.~144]{hopcroftUllman1967}, see also~\cite[Lemma 12.4]{bojanczyk_automata_2018}, which shows that two-way transducers are closed under pre-composition with Mealy machines. This construction is used, for example, to show that two-way transducers are closed under composition, see~\cite[p.~139]{chytilSerialComposition2Way1977}. 
 As shown in this section, thanks to the Krohn-Rhodes theorem, the tree-trimming construction is not needed, since one  only needs to show that two-way transducers are closed under pre-composition with the prime functions, which is a simpler undertaking.
\subsection{Compositions of two-way primes to regular list functions}
\label{ap:composition-to-regular-list-functions}

    In this section we show  implication \ref{two-way:krohn-rhodes} $\Rightarrow$ \ref{two-way:regular-list-function} in Theorem~\ref{thm:two-way-models}. This implication can be stated as:
\begin{align*}
    \twoinclusion
{(two-way primes)*} {regular list functions with atoms.}
\end{align*}
For the purposes of this section, we use the name \emph{derivable} as a synonym for regular list function with atoms. Since the derivable functions are closed under composition by definition, it is enough to show that every two-way prime is derivable. 

String-to-string homomorphisms will be handled in Section~\ref{sec:derivable-homomorphisms}. 
Map reverse and map duplicate were already shown to be derivable in Example~\ref{ex:map-reverse}.  The remaining derivable functions are of the form  $f|id$ where $f$ is either atom propagation, a group transducer, or the flip-flop transducer; these are discussed in  Section~\ref{sec:derive-atom-propagation}.

\subsubsection{String-to-string homomorphisms}
\label{sec:derivable-homomorphisms}
We begin with equivariant string-to-string homomorphisms. The main observation is the following lemma, which says that every equivariant function with a polynomial orbit-finite domain is necessarily derivable. In other words, if the function is equivariant and its domain does not use the list datatype constructor (because such domains are exactly the polynomial orbit-finite sets), then it is derivable.

\begin{lemma}\label{lem:derive-orbit-finite-domain}
    If  $f : \Sigma \to \Gamma$ is an equivariant function between two datatypes, and the domain  $\Sigma$  is a polynomial orbit-finite set, then $f$  is derivable. 
\end{lemma}
Before proving the lemma, we use it to show that every equivariant string-to-string homomorphism $f : \Sigma^* \to \Gamma^*$ is derivable, assuming that the input and output alphabets $\Sigma$ and $\Gamma$ are polynomial orbit-finite sets. The  homomorphism is obtained by lifting some equivariant function of type $\Sigma \to \Gamma^*$  (which is derivable  thanks to the lemma) to strings via the {\tt map} combinator from Definition~\ref{def:reglist-fun}, and then applying the function {\tt concat} from Figure~\ref{fig:atomic-combinators-1} to the result.  Therefore, the homomorphism is derivable.
\begin{proof}
    In the proof it will be more convenient to work with the following representation of polynomial orbit-finite sets.

    \begin{claim}
        Every polynomial orbit-finite set admits a  bijection (derivable in both directions) with a set of the form 
        \begin{align*}
            \atoms^{k_1} + \cdots + \atoms^{k_n} \qquad\text{for some $n, k_1,\ldots,k_n \in \set{0,1,\ldots}$}.
            \end{align*}
    \end{claim}
    \begin{proof}
        Using distributivity of product across co-product, and the fact that  singleton sets are of the form $\atoms^0$ (up to derivable bijections with sets). Furthermore, this construction is derivable, since it can be formalised using the operations from Definition~\ref{def:reglist-fun}, in particular the distributivity function from Figure~\ref{fig:atomic-combinators-1}.
    \end{proof}
    
    In order to prove the lemma, it is enough to prove it for the case where  $\Sigma$ is  of the form given in the claim. Furthermore, the co-products in the  domain $\Sigma$ can be safely ignored, since every equivariant function 
        \begin{align*}
        f : \Sigma_1 + \Sigma_2 \to \Gamma
        \end{align*}
        can be obtained from two equivariant functions $\set{f_i : \Sigma_i \to \Gamma}_{i = 1,2}$  by combining them using the {\tt cases} combinator from item~\ref{regfun:combinators} in Definition~\ref{def:reglist-fun}. Therefore, it is enough to prove the lemma for the case when $\Sigma$ is $\atoms^k$ for some $k \in \set{0,1,2,\ldots}$.  Recall that by orbits of $\atoms^k$, we mean  orbits with respect to the action of  atom automorphisms. The following claim shows that the partition of $\atoms^k$ into orbits is derivable. 
        \begin{claim}\label{claim:derivable-partition-into-orbits}
            For every orbit  $\tau \subseteq \atoms^k$,
            its characteristic function $\tau : \atoms^k \to \boolset$ is a regular list function with atoms.
        \end{claim}
        \begin{proof}
            Tuples in $\atoms^k$ are in the same orbit if and only if they have the same equality types, see~\cite[Lemma 7.5]{bojanczyk_slightly2018}. Using the equality function on atoms, which is one of the prime derivable functions,  we can compute for each pair of coordinates in $\set{1,\ldots,k}$ whether or not the corresponding atoms are equal,
            obtaining the equality pattern.
            The Claim follows.
        \end{proof}
        Derivable functions admit a conditional construction~\cite[Example 2]{bojanczykRegularFirstOrderList2018}: if $f : \Sigma \to \boolset$ is derivable, and $\set{g_i : \Sigma \to \Gamma}_{i \in \boolset}$ are derivable, then the same is true for
        \begin{align*}
        a \in \Sigma \quad \mapsto \quad g_{f(a)}(a).
        \end{align*}
        Using this construction, and the derivable partition into orbits from Claim~\ref{claim:derivable-partition-into-orbits}, the lemma will follow once  we  prove the following claim.
    
        \begin{claim}
            Let $\Gamma$ be a datatype, and let  $f : \atoms^k \to \Gamma$ be equivariant. For every orbit $\tau \subseteq \atoms^k$, there is a derivable function of type $\atoms^k \to \Gamma$ which agrees with $f$ on arguments from $\tau$. 
        \end{claim}
        \begin{proof}
            Induction on the structure of the  output datatype $\Gamma$. More formally, the induction is on (a) the number of list datatype constructors; followed by (b) the number of other datatype constructors. These parameters are ordered lexicographically.
            \begin{enumerate}
                \item Suppose that  $\Gamma$ is a singleton set.  In this case, we can use item~\ref{regfun:constant} of Definition~\ref{def:reglist-fun}.
                \item Suppose that $\Gamma$ is the atoms $\atoms$.   Every  equivariant function of type $\tau \to \atoms$ must necessarily be a projection function. This is because if an equivariant function agrees with the $i$-th projection for some argument $a \in \tau$, then it must agree with the $i$-th projection for all other arguments in the  orbit $\tau$. Furthermore, projections are derivable,  see Figure~\ref{fig:atomic-combinators-1}.
                \item Suppose that $\Gamma$ is a co-product $\Gamma = \Gamma_1 + \Gamma_2$. If two inputs of $f$ are in the same orbit, then the corresponding outputs will be either both in $\Gamma_1$ or in $\Gamma_2$, and therefore  all outputs (assuming the inputs are in a fixed orbit $\tau$) will be from some fixed $\Gamma_i$. So we can use the induction assumption.
                \item Suppose that $\Gamma$ is a product $\Gamma = \Gamma_1 \times \Gamma_2$. Then $f$ can be obtained by applying the {\tt pair} combinator to its two projections onto $\Gamma_1$ and $\Gamma_2$.
                
                \item The most interesting case is when $\Gamma$ is a list datatype $\Gamma = \Delta^*$. The crucial observation is that, when  restricted to inputs from a fixed orbit $\tau \subseteq \atoms^k$, all outputs of $f$ have the same length (as lists). This is  because list length is invariant under atom automorphisms.       
                Let $m \in \set{0,1,\ldots}$ be the list length of all outputs from $f(\tau)$.
                Define 
                 \begin{align*}
                 g : \atoms^k \to  \Gamma^m+ \bot
                 \end{align*}
                 to be the function which gives $\bot$ for arguments outside $\tau$, and the same output as $f$ for arguments from $\tau$. In the latter case, the output is represented not as a list, but as  $m$-tuple of  $\Gamma$. The function $g$ is equivariant and hence derivable by induction assumption (there are fewer list datatype constructors in the co-domain type). To derive $f$, we combine $g$ with the natural embedding $\Gamma^m \to \Gamma^*$, 
                 which can be derived using the {\tt append} function. 
            \end{enumerate}
        \end{proof}
\end{proof}

\subsubsection{Atom propagation, group transducers, and flip-flop}
\label{sec:derive-atom-propagation}
It remains to show derivability for functions which are  a parallel composition  of the  form  $f|id$ where $f$ is either atom propagation, a group transducer, or the flip-flop transducer, and $id$ is the identity strings over some polynomial orbit-finite alphabet. 

When $f$ is a group transducer, then $f|id$ is simply the last function from Figure~\ref{fig:atomic-combinators-1}. Atom propagation and flip-flop are done in similar ways, so we only discuss atom propagation, which is the most interesting example.  Also, we only show how atom propagation is itself derivable, and  leave the reader to lift the construction via parallel composition with an identity homomorphism.

To see the construction, we recall the example of atom propagation from Example~\ref{ex:atom-propagation}:
\begin{align*}
    \begin{array}{rccccccccccccccccccc}
        \text{input} & 1&2&\propignore & \propignore & \propload & \propload & 3 & \propignore & \propignore & \propload & \propignore & \propload \\
        \text{output} & \bot & \bot & \bot & \bot & 2 &  \bot & \bot & \bot & \bot & 3 & \bot & \bot
    \end{array}
    \end{align*}
Consider an input to atom propagation, e.g.~the input from the above example
\begin{align*}
[1,2,\varepsilon,\varepsilon,\downarrow,\downarrow, 3,\varepsilon,\varepsilon,\downarrow,\varepsilon,\downarrow] \in (\atoms + \set{\varepsilon,\downarrow})^*
\end{align*}
Apply the {\tt block} operation, with the two kinds of letters being $\downarrow$ and the remaining letters. The result looks like this:
\begin{align}\label{eq:after-block}
    [\red{[1,2,\varepsilon,\varepsilon]},[\downarrow,\downarrow],[3,\varepsilon,\varepsilon],[\downarrow],[\varepsilon],[\downarrow]] \in ( (\atoms + \varepsilon)^* + \downarrow^* )^*
    \end{align}
(We have coloured the first block red, for reasons that will be explained below.)
Formally speaking, in order to apply the {\tt block} operation above, we need to refactor the alphabet via the obvious bijection
\begin{align}\label{eq:string-partitioned-downarrows}
(\atoms + \set{\varepsilon,\downarrow}) \to ((\atoms+\varepsilon) + \downarrow),
\end{align}
which is derivable thanks to Lemma~\ref{lem:derive-orbit-finite-domain}. To the list from~\eqref{eq:after-block}, apply the windows function
\begin{align*}
[x_1,\ldots,x_n] \mapsto [(x_1,x_2),(x_2,x_3),\ldots,(x_{n-1},x_n)],
\end{align*}
which is derivable thanks to~\cite[Example 3]{bojanczykRegularFirstOrderList2018}. The result looks like this:
\begin{align}\label{eq:list-with-windows}
    [(\red{[1,2,\varepsilon,\varepsilon]},[\downarrow,\downarrow]),([\downarrow,\downarrow],[3,\varepsilon,\varepsilon]),([3,\varepsilon,\varepsilon],[\downarrow]),([\downarrow], [\varepsilon]),([\varepsilon],[\downarrow])] \in ( ((\atoms + \varepsilon)^* + \downarrow^*)^2 )^*
\end{align}
Using the {\tt map} combinator, to each  element 
\begin{align*}
x \in ((\atoms + \varepsilon)^* + \downarrow^*)^2
\end{align*}
of the above list, apply the following operation (whose derivability is left to the reader):
\begin{itemize}
    \item If the second coordinate of $x$ is in  $(\atoms + \varepsilon)^*$, then project $x$ to the second coordinate, and then replace every letter by $\bot$, as in the following example:
    \begin{align*}
        ([\downarrow,\downarrow],[3,\varepsilon,\varepsilon]) \mapsto [\bot,\bot,\bot]
    \end{align*}
    \item Otherwise, if the second coordinate of $x$ is in $\downarrow^*$, then do the same thing as in the previous item, except that the first $\bot$ is replaced by the  last atom in the first coordinate of $x$,  which might be the undefined value $\bot$ if there is no such atom. Here is an example:
    \begin{align*}
        ([1,2,\varepsilon,\varepsilon],[\downarrow,\downarrow]) \mapsto [2,\bot]\\
        ([\varepsilon,\varepsilon],[\downarrow,\downarrow]) \mapsto [\bot,\bot].
    \end{align*}
    To take out the last atom in the first coordinate of $x$, we use a derivable filter operation~\cite[Example 1]{bojanczykRegularFirstOrderList2018}, followed by {\tt reverse} and {\tt coappend}. 
\end{itemize}
After applying the above operation to every $x$ in the list from~\eqref{eq:list-with-windows}, we get a result that looks like this:
\begin{align*}
    [[2,\bot],[\bot,\bot,\bot],[3],[\bot],[\bot]] \in ((\atoms + \bot)^* )^*.
\end{align*}
Applying {\tt concat}, yielding:
\begin{align*}
    [2,\bot,\bot,\bot,3,\bot,\bot] \in (\atoms + \bot)^* .
\end{align*}
This is almost the same as the output of atom propagation, but with one difference -- we have omitted the output for the first block of letters that was coloured red in~\eqref{eq:after-block}. This block will produce a sequence of $\bot$ values -- of length equal to the red block -- and this sequence can be prepended to the output in a derivable way. 
\subsection{Compositions of two-way primes to streaming string transducers}
\label{ap:composition-to-sst}
In this section we show the implication \ref{two-way:krohn-rhodes} $\Rightarrow$ \ref{two-way:sst} in Theorem~\ref{thm:two-way-models}, which can be stated as:
\begin{align*}
    \twoinclusion
{(two-way primes)*} {\sst with atoms.}
\end{align*}
To prove the above inclusion, we show that \sst with atoms are closed under post-composition with the two-way primes:
\begin{align*}
    \threeinclusion{ (two-way primes)}{(\sst with atoms)} {\sst with atoms}
\end{align*}
Note the difference with Section~\ref{ap:composition-to-two-way}, where pre-composition was used instead of post-composition. The  inclusion is proved in the following lemma.

\begin{lemma}\label{lem:post-compose} If $f$ is an \sst with atoms, and $g$ is a prime two-way function, then $g \circ f$ is  an \sst with atoms.
\end{lemma}
\hideproof{
\begin{proof}
    There are five cases, depending on which  prime two-way function is used for $g$.
\begin{itemize}
    \item \emph{Map reverse.} Suppose that $g$ is map reverse.  In the transducer for $g \circ f$, each string  register of  $f$ is replaced by three string registers, which store the contents before the first separator, between the first and last separators, and after the last separator. This is illustrated the following picture: 
    \mypic{7}
    Furthermore, the  \sst with atoms for $g \circ f$  remembers in its state if the middle register $A_2$ is empty, i.e.~if there is a separator. (If there is no separator, then all of the letters are in $A_1$.)
    An action $A := BC$  is simulated by
    \begin{align*}
            A_1 := B_1 \quad A_2 := B_2 C_1 B_3 C_2 \quad A_3 := C_3
    \end{align*}
    if register $C_2$ is nonempty, and otherwise it is simulated by 
    \begin{align*}
    A_1 := B_1 \quad A_2 := B_2  \quad A_3 := C_1 B_3.
    \end{align*}
    \item \emph{Map duplicate.}  Same idea as above, except that five registers used:
    \mypic{6}

    \item \emph{String-to-string homomorphisms.} Instead of a writing a letter into a string registers, one writes its homomorphic image.
    \item We now consider the case when  $g$ is either atom propagation, a group transducer, or the flip-flop transducer. (Formally speaking, we should consider the case when $g$ is a parallel composition of one of the transducers mentioned above with an identity homomorphism, but the construction works the same way.) 
    \begin{enumerate}
        \item \emph{Atom propagation.} Assume that $g$ is atom propagation. In the transducer $g \circ f$,  each string register $A$ of $f$ is replaced by two string registers $A_1,A_2$ and one atom register $a$.  The two string registers store the output of atom propagation on $A$, split into the parts that are  before and after the first $\downarrow$, with neither part including the first $\downarrow$. The atom register stores  the last atom in $A$. Here is a picture:
    \mypic{8}
     An action $A := BC$ is simulated by 
    \begin{align*}
    A_1 := B_1 \quad A_2 := B_2 C_1 b C_2  \quad a:=c.
    \end{align*}
        \item \emph{A group transducer.} Assume that $g : G^* \to G^*$ computes the group product for each prefix, for some finite group $G$. For $q \in G$, define $g_q :G^* \to G^*$ to be the variant of $g$ obtained  by changing the initial state to $q$ (i.e.~the $i$-th letter of the output is $q$ times the product of the first $i$ letters of the input).
     
        In this case we can apply a natural construction that stores images of a string register under all functions $g_q$ (this natural construction will respect the single-use restriction because we are dealing with a group): In the transducer $g \circ f$, each string register $A$ of $f$ is replaced by a family of string registers $\set{A_q}_{q \in G}$.  The invariant is that  $A_q = g_q(A)$. Also, in its state, the transducer $g \circ f$ stores the group product of each string register. An action $A:=BC$ in the original transducer is replaced by  
        \begin{align*}
        A_q := B_q C_{q \cdot  \text{(group product of $B$)}} \qquad \text{for every $q \in G$}.
        \end{align*}
        The key observation is that the above updates are single-use, because multiplying by an element of a group induces a permutation of the group. 
        \item \emph{The flip-flop monoid.} The construction is similar to the one for atom propagation. A string register $A$ of $f$ is replaced by three string registers, defined as follows. For $\sigma \in \set{a,b}$, register $A_\sigma$ stores the output of the flip-flop on the prefix of $A$ up to and including the first non-identity letter, assuming that the initial state of the flip-flop is $\sigma$. String register $A'$ stores the output of the flip-flop on the part of $A$ after the first non-identity letter; this part of the output does not depend on the initial state. Here is a picture:
        \mypic{9}
         An action $A:=BC$ in the original transducer is simulated by 
        \begin{align*}
        A_a := B_a \qquad A_b := B_b \qquad A' := B' C_\sigma C'
        \end{align*}
        where $\sigma \in \set{a,b}$ is the last non-identity letter that was used in register $B$ (this letter is stored in the state of the transducer). If $\sigma$ is undefined, because $C$ used only identity letter, then $A:=BC$ is simulated by 
        \begin{align*}
            A_a := B_a C_a \qquad A_b := B_b C_b \qquad A' :=  C'.
            \end{align*}
             
    \end{enumerate}
    
\end{itemize}

\end{proof}
}

This completes the proof that compositions of two-way primes are contained in \sst with atoms. Before continuing, we observe a corollary of the above proof.

Recall that a   \emph{multiple-use} \sst with atoms is the variant of an \sst with atoms  where the single-use restriction is lifted, both for atom registers and for string registers (we discussed this model in Theorem~\ref{thm:non-equivalent-copyful}). 
The proof of Lemma~\ref{lem:post-compose} also works when $f$ is multiple-use  and therefore
we obtain the following inclusion
\begin{align}\label{eq:copyful-sst-with-atoms-primes}
    \threeinclusion{(two-way primes)*} {(multiple-use \sst with atoms)}
    {multiple-use \sst with atoms}
    \end{align}
Once we prove Theorem~\ref{thm:two-way-models}, we will establish that compositions of two-way primes are the same as (single-use) \sst with atoms, which gives the following result:
\begin{align}\label{eq:copyful-sst-with-atoms}
    \threeinclusion{(\sst with atoms)} {(multiple-use \sst with atoms)}
    {multiple-use \sst with atoms}
    \end{align}
In the above inclusion, it is important that use post-composition with  \sst with atoms, and not pre-composition. The  inclusion for pre-composition 
\begin{align}\label{eq:false-inclusion}
    \threeinclusion{(multiple-use \sst with atoms)}{(\sst with atoms)} 
    {multiple-use \sst with atoms}    
\end{align}
does not hold. A counterexample is the function $f$  that  gives all suffixes of the input word (ordered from shortest to longest), e.g.
    \begin{align*}
    12345 \mapsto 554543543254321.
    \end{align*}
The function belongs to left side of the inclusion~\eqref{eq:false-inclusion}, because it can be obtained by first reversing the input string, and then producing all prefixes of the result. On the other hand, $f$ does not belong to right side of the inclusion, since every function $f$ computed by a multiple-use \sst with atoms with $k$ string registers satisfies the following invariant:
\begin{itemize}
    \item[(*)] For every input $w \in \atoms^*$, in the output $f(w)$ there are at most $2k$ letters that appear in positions adjacent to some appearance of the last letter of  $w$. 
\end{itemize}
This invariant is violated by our function $f$, for every $k$, and hence $f$ is not computed by a multiple-use \sst with atoms.

The inclusion~\eqref{eq:copyful-sst-with-atoms} is also valid -- with the same proof but without the need for atom propagation -- in the case without atoms, where the term word \emph{copyful} is used instead of \emph{multiple-use}, see~\cite{filiotCopyfulStreamingString2017}:
\begin{align}\label{eq:copyful-sst-without-atoms}
\threeinclusion{(\sst without atoms)} {(copyful \sst without atoms)}
{copyful \sst without atoms}
\end{align}
    which may be a result of independent interest. In fact, we use this result in the proof of Theorem~\ref{thm:equivalence} about decidability of equivalence.

\subsection{Two-way transducers to compositions of two-way primes}
\label{ap:two-way-to-composition}
In this section we show the implication \ref{two-way:two-way} $\Rightarrow$ \ref{two-way:krohn-rhodes} in Theorem~\ref{thm:two-way-models}, i.e.~the inclusion
\begin{align*}
    \twoinclusion
{single-use two-way transducers} {(two-way primes)*}
\end{align*}
This part of the proof can be seen as a Krohn-Rhodes Theorem form two-way transducers. Fortunately, half  of the work for this part has already been done previously, when proving the Krohn-Rhodes Theorem for one-way transducers.

For the rest of this section fix a  single-use two-way transducer. 
For an input string, define its \emph{run graph} to be the directed graph where the vertices are extended states that appear in the accepting run, the edges connect consecutive extended states, and each vertex is labelled by the output that its extended state produces -- for extended states which execute an output action, the label is a single output letter, and for the remaining extended states the output is the empty string $\varepsilon$.  Here is a picture of a run graph (the input string is blue, the run graph is red) for the map duplicate  transducer from Example~\ref{ex:map-reverse-duplicate}:
\mypic{19}    
For transducers which input and output run graphs, 
we represent the  run graph as a string\footnote{This string is called the history of a run in~\cite[p.~137]{chytilSerialComposition2Way1977}}, where the letters are the columns of the picture, as explained in the following picture
 \mypic{18}

In order to have an orbit-finite alphabet for this representation, we need to show that the run graph has a bounded number of rows, i.e.~every input position is visited a bounded number of times. This is done in the following lemma.
\begin{lemma}\label{lem:bounded-width}
There is some $k \in \set{1,2,\ldots}$ such that every accepting run of the fixed two-way single-use  transducer  visits every input position  at most $k$ times.
\end{lemma}
\begin{proof}
    Without atoms, the  lemma is obvious, since a non-looping run can visit each position  at most once in a given state. With atoms, the lemma crucially relies on the single-use restriction, since there exist multiple-use two-way automata which visit positions an unbounded number of times. An example  is the two-way automaton~\cite[Example 11]{kaminskiFiniteMemoryAutomata1994} that checks if some letter appears twice.
    
    Consider a distinguished position inside an input string, as in the following picture:
    \mypic{17}
    Consider the  Shepherdson functions of the ``before'' and ``after'' parts. By Lemma~\ref{lem:small-support}, there is a tuple of atoms $\bar a \in \atoms^*$  which supports both of these Shepherdson functions, and  which has  bounded size, i.e.~the length of the tuple $\bar a$ depends only on the fixed transducer. Consider the  configurations $c_1,\ldots,c_n$ in the run of the transducer where the head is over the distinguished position. The extended state $c_{i+1}$ can be obtained from the extended state $c_i$ by using the Shepherdson functions for the ``before'' and ``after'' parts as well as the letter in the distinguished position; it follows that all of the configurations $c_1,\ldots,c_n$ are supported by $\bar a$ plus the atoms that appear in the first extended state $c_1$ and in the label of the  distinguished position. In other words, every atom that appears in a extended state $c_i$ must appear either in $\bar a$, in the register valuation from $c_1$, or in the distinguished position. There is a bounded number of register valuations that can be constructed using a bounded number of atoms, hence the lemma  follows. 
\end{proof}
Thanks to Lemma~\ref{lem:bounded-width}, run graphs can be viewed as strings over a polynomial  orbit-finite alphabet, namely
\begin{align*}
\underbrace{\text{($\set \varepsilon\ +$ output alphabet)}}_{\text{label of vertex}} \qquad \times \qquad \big(  \underbrace{(\set{\text{-1,0,1}}}_{\substack{\text{column offset of}\\ \text{the next vertex}}} \times \underbrace{\set{1,\ldots,k})}_{\substack{\text{row of the}\\ \text{next vertex}}} \quad  + \quad  \underbrace{\bot}_{\substack{\text{no next}\\\text{vertex}}} \big).
\end{align*}
Using this representation, we below  show that the following transformations can be computed by compositions of prime two-way functions:
\begin{align*}
\underbrace{\text{input string} \mapsto \text{run graph}}_{\text{Section~\ref{sec:from-input-string-to-run-graph}}} \qquad \underbrace{\text{run graph} \mapsto \text{output string}}_{\text{Section~\ref{sec:from-run-graph-to-output-string}}}.
\end{align*}
This will complete the proof that every single-use two-way transducer is a composition of prime two-way functions.

\subsubsection{From a run graph to the output string}
\label{sec:from-run-graph-to-output-string}
We begin by showing that the transformation
\begin{align*}
    \text{run graph} \mapsto \text{output string}
\end{align*}is a composition of two-way primes. Define the \emph{width} of a run graph to be the maximal number of times a position is visited. We prove the result by induction on the width: i.e.~for every $k \in \set{1,2,\ldots}$, we will show that there is a composition of two-way primes which computes the output string for run graphs of width at most $k$.  This is enough, because by  Lemma~\ref{lem:bounded-width}, the width of run graphs is bounded.

All constructions in this section assume that the input run graph is  a single path, i.e.~the run graph is connected, and every node has in-degree and out-degree at most one. Here is a non-example:
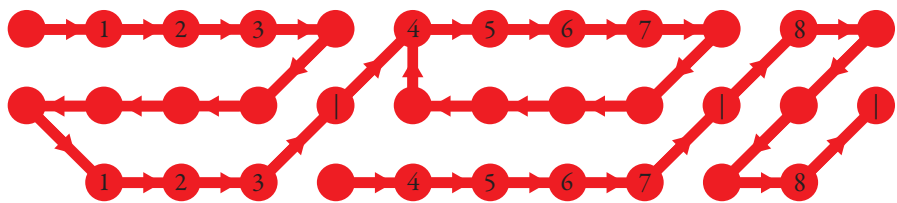
\mypic{5}
Whether or not a run graph is a single path can be checked by a single-use automaton. 

\smallparagraph{Induction base} The induction base says that the output string of a run graph of width $k=1$ can be computed by a composition of two-way primes.  There are two kinds of run graphs of width $1$, shown  in the following picture:
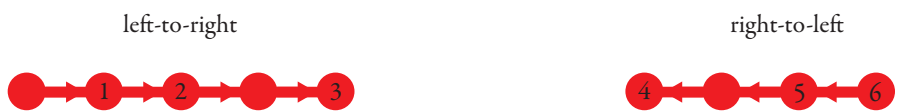
\mypic{20}
In the left-to-right  case, the output word is obtained by applying a homomorphism. In the right-to-left case, we need to first reverse the input string (reversing is a special case of map reverse, when no separators are used). Both cases are therefore compositions of two-way primes, and the case disjunction is handled using the following lemma. 
\begin{lemma}\label{lem:if-then-else}
    Let $L \subseteq \Sigma^*$ be a language recognised by single-use two-way automaton. If $
    f_1, f_2 : \Sigma^* \to \Gamma^*$ are  compositions of  two-way primes, then the same is true for 
    \begin{align*}
    w \in \Sigma^* \qquad \mapsto \qquad \begin{cases}
        f_1(w) & \text{if $w \in L$}\\
        f_2(w) & \text{otherwise}.
    \end{cases}
    \end{align*}
\end{lemma}
\begin{proof}
    Consider two copies of $\Sigma$, a black copy $\Sigma$ and a red copy $\red \Sigma$. Consider the function
    \begin{align*}
    \Sigma^* \to (\Sigma + \red \Sigma)^*
    \end{align*}
    which colours all positions black if the input belongs to $L$, and red otherwise. This function is a composition of the following functions:
    \begin{enumerate}
        \item append a letter ``yes'' or ``no'' to the string, depending on whether the string belongs to $L$;
        \item reverse the string;
        \item colour all positions black if the first letter is ``yes'', red otherwise; 
        \item remove the ``yes''/``no'' letter;
        \item reverse the string again.
    \end{enumerate}
    Apart from reversal, all of the above operations are single-use one-way transducers, which can be decomposed into two-way primes thanks to the Krohn-Rhodes decomposition from  Theorem~\ref{thm:mealy-machine-krohn-rhodes}. To finish the job, we apply sequentially functions 
    \begin{align*}
    \xymatrix{
        (\Sigma + \red \Sigma)^* \ar[r]^{g_1} &
         (\Gamma + \red \Sigma)^* \ar[r]^{g_2}  &
         (\Gamma + \red \Gamma)^*
    },
    \end{align*}
    where $g_1$ applies $f_1$ if the input uses black letters and is the identity otherwise, while $g_2$ applies $f_2$ if the input uses red letters and is the identity otherwise. These can be easily shown to be compositions of two-way primes.  Finally, we use a homomorphism to ignore the distinction between black and red letters. 
\end{proof}

\smallparagraph{Induction step} Let $k \in \set{2,3,\ldots}$. Assume that we can compute -- using a composition of two-way primes -- the output string for every run graph of width $<k$. We will now show that the same is true  for width $k$. In Section~\ref{sec:loop-case}, we prove the induction step for a special kind of run graphs called loops, and then in Section~\ref{sec:sweeps-general-case} we prove the general case.

\subsubsection{Loops}
\label{sec:loop-case}
A run graph is called a \emph{loop} if its first and last configurations are in the same position.
 In this section we  show  how the output string can be computed for run graphs which are loops.  Because the set of loops is recognised by a single-use  automaton (even one-way), and because we have the conditional construction from Lemma~\ref{lem:if-then-else}, in the following construction we do not need to worry what happens when the input is not a loop.

To make the notation easier, consider first the special  case of a \emph{right loop}, which is a loop that only visits positions to the right of its first position. 
    The idea is to decompose a right loop into two parts as follows:
    \mypic{23}
    Both parts have smaller width. Therefore,  it is enough to show that a composition of two-way primes can compute the decomposition, which is done in the following lemma. 
    \begin{lemma}\label{lem:rational-loop}
        There is a composition of two-way primes, such that if the input is a right loop, then the  output is a concatenation of (the strings representing) the run graphs of the first and second parts. The output strings are separated by a separator symbol $|$.
    \end{lemma}
    \begin{proof}
        A classical result in transducer theory (over finite alphabets) is a theorem of Elgot and Mezei which says  that every function computed by an unambiguous nondeterministic one-way transducer can be computed by a two-pass process: first a left-to-right one-way deterministic transducer, followed by a right-to-left one-way deterministic transducer, see~\cite[Proposition 7.4]{elgot_relations_1965}. If we consider the same models for polynomial orbit-finite alphabets (but finite state spaces and no registers), then the same result carries over, with the appropriate deterministic model being  the class
        \begin{align}\label{eq:single-use-rational}
            \small
        (\text{homomorphism} + \underbrace{\text{reverse} \circ \text{(single-use Mealy)} \circ \text{reverse}}_{\text{right-to-left single-use Mealy}} + \text{(single-use Mealy)})^*.
        \end{align} 
        We use the name  \emph{single-use rational function} for functions from the above class. Thanks to the Krohn-Rhodes decomposition for single-use Mealy machines from Theorem~\ref{thm:mealy-machine-krohn-rhodes}, every single-use rational function is a composition of two-way primes.
        
        We use  single-use rational functions to compute the function in the lemma. Consider first the function which inputs a run graph that is a right loop, and outputs the same run graph with every node coloured either yellow or blue, depending on whether it is in the first or second part. This function can be computed by an unambiguous nondeterministic device (without any registers), which guesses the partition into yellow and blue positions. Therefore,   by the Elgot and Mezei construction, it is a single-use rational function
        (this is because the transducer does not use any registers). Once we have the configurations coloured yellow or blue, we can easily complete the proof of the lemma:   use map duplicate to create two consecutive copies of the run graph, separated by a separator symbol, and then a  single-use Mealy machine to  keep only the yellow configurations in the first copy, and only the blue configurations in the second copy.
    \end{proof}
    As mentioned previously, the first and second parts of a right loop have smaller width, and therefore their outputs can be computed using the induction assumption. To apply the induction assumption on both sides of the separator in the output of the transducer from the above lemma,  we use the map combinator that was defined in Section~\ref{sec:closure-properties}.

    \begin{lemma}\label{lem:two-way-map}
        If $f$ is a composition of two-way primes, then the same is true for  $\map f$.
    \end{lemma}
    \begin{proof}
        We use the name \emph{map $f$} for the operation defined in the lemma. 
        Since map commutes with sequential composition, it is enough to show the lemma for the two-way primes. Clearly, for every single-use Mealy machine its map is also a single-use Mealy machine, and therefore a composition of two-way primes. 
        This covers all prime two-way functions except for map duplicate and map reverse.  Suppose want to apply 
        \begin{align*}
        \overbrace{\text{map} (\underbrace{\text{map reverse}}_{\text{\red{red separator $|$}}})}^{\text{black separator $|$}}
        \end{align*}
        This is the same as applying map reverse with both kinds of separators being treated as the separators, which can be seen in the following example: 
        \begin{align*}
        0|12\red|34\red|56|7\red|89 \quad \mapsto \quad  0|21\red|43\red|65|7\red|98
        \end{align*}
        This operation can be easily be done by composition of two-way primes. The argument for map duplicate is the same. 
    \end{proof}
    The general case of loops -- instead of just right loops -- is treated in a similar way, except that we decompose the loop into a constant number of loops, each one of which is a right loop or a left loop.

    \subsubsection{The general case}
    \label{sec:sweeps-general-case}
    In this section  we show how a composition of two-way primes can  compute the output string for any  run graph of width $k$ (not necessarily a loop). The idea is to decompose the run into loops and parts which connect them.

    Assume that the last position in the run is to the right of the first position. The other case is treated symmetrically, and the distinction between the cases is handled using  the conditional construction from Lemma~\ref{lem:if-then-else}.  

    For $i \in \set{1,2,\ldots}$, define the  \emph{$i$-th station} (which is a position in the input string) and the \emph{$i$-th loop and sweep} (which are parts of the run) as follows. The first station is the starting position of the run. Suppose that we have already defined the $i$-th station.  Define  \emph{$i$-th loop}  to be the  part of the run that begins with the first visit in the $i$-th station, and ends in the last visit there. Define  \emph{$(i+1)$-st station} to be the
    the first position to the right of the $i$-th station, that was not visited by the $i$-th loop. Finally, define the \emph{$i$-th sweep} to be the part of the run that connects stations $i$ and $i+1$ (for the first time). These definitions are illustrated below:
 \mypic{24}
   
 Using the decomposition into loops and sweeps, the output string of the run graph is computed in three steps, as described below.
    \begin{enumerate}
        \item Colour each node in the run either yellow or blue, 
         depending on whether this node is part of a loop or part of sweep. This colouring is a single-use rational function, using the Elgot and Mezei construction discussed in the proof of  Lemma~\ref{lem:rational-loop}: the partition into loops and sweeps can be computed by an unambiguous nondeterministic automaton, that
         does not use registers. 
         \item Define $i$-th window to be the maximal interval in the run which contains the $i$-th station but no other stations. The windows are overlapping. Here is a picture:\mypic{25} 
         The important property of windows is  that the $i$-th sweep and the $i$-th loop are both contained entirely in the $i$-th window. Transform the output from the previous step into the concatenation of windows, as illustrated in the following picture:
         \mypic{26}
        This transformation can be done by composition of two-way primes by placing separators in the stations, and using map duplicate and a similar idea to the windows construction in~\cite[Example 3]{bojanczykRegularFirstOrderList2018}.
        \item For each window (the iteration over windows is possible thank to the map construction from Lemma~\ref{lem:two-way-map}), compute the output of the unique loop  that is entirely contained in that window and the sweep that follows it. Both the loop and the sweep can be isolated, using a single-use Mealy machine, thanks to the yellow and blue colours from the first step. The output string for the loops can be produced using the construction from Section~\ref{sec:loop-case}. The output string for the sweeps can be produced by induction assumption on smaller width, since every position visited by the $i$-th sweep is also visited by the $i$-th loop, and therefore the $i$-th sweep has smaller width than the entire run.
    \end{enumerate}
    
This completes the proof that the output string of a run graph can be computed by a composition of two-way primes.

\subsubsection{From an input string to a run graph}
\label{sec:from-input-string-to-run-graph}

In this section we show that the function
\begin{align*}
 \textrm{input string} \mapsto \textrm{run graph}
\end{align*}
is a composition of two-way primes. We do it by showing that it is a single-use rational function, 
which is a stronger statement.
We start the construction by applying the Split Lemma
for the input word
for the monoid morphism $h$ that maps a word $w$ to its
Shepherdson profile as defined in Section \ref{sec:automata-to-semigroups}. Recall that the Shepherdson profile is a function of the type:
\begin{align*}
     \overbrace{Q \times (\atoms+\bot)^k}^{\substack{\text{state and register}\\ \text{valuation at the}\\ \text{start of the run}}}
     \times \overbrace{\set{\leftarrow, \rightarrow}}^{\substack{\text{does the run}\\ \text{enter from the}\\ \text{left or right}}}  
     \qquad \to  \qquad \set{\text{accept, loop}} +  (\overbrace{Q \times (\atoms+\bot)^k}^{\substack{\text{state and register}\\ \text{valuation at the}\\ \text{end of the run}}} \times \overbrace{\set{\leftarrow, \rightarrow}}^{\substack{\text{does the run}\\ \text{exit from the}\\ \text{left or right}}})
\end{align*}
As mentioned in Section~\ref{sec:automata-to-semigroups} Shepherdson profiles form an orbit-finite monoid, call it $M$. 
Similarly as in Section~\ref{sec:Mealy-transformation-monoid} we define the set of extended states
\begin{align*}
\bar Q = Q \times (\atoms+\bot)^k \times \set{\leftarrow, \rightarrow} \;\;+\;\; \set{\lightning, \checkmark, \xmark}
\end{align*}
where $\lightning$ denotes an error in computation
due to accessing undefined register value,
$\xmark$ denotes looping or rejecting,
and $\checkmark$ denotes a computation that has ended successfully. Define set $\overrightarrow Q \subseteq Q$
of all the extended states that carry the $\rightarrow$ value (similarly define the set $\overleftarrow Q$). 
There are natural left and right actions of $M$ on $\bar Q$.
For every $\bar q \in \bar Q$ and a $m \in M$
\begin{itemize}
    \item $\bar q m = m(\bar q)$ if $\bar q \in \overrightarrow Q$, or $\bar q m = \bar q$ otherwise;
    \item $m \bar q = m(\bar q)$ if $\bar q \in \overleftarrow Q$, or $m \bar q = \bar q$ otherwise.
\end{itemize}
Define left- and right- compatibility just like in Section \ref{sec:Mealy-transformation-monoid}:
e.g. say that $\bar q \in \bar Q$ is left-compatible with $m \in M$ if $m \bar q \neq \lightning$.
We now prove a couple of lemmas that show, how the two-way transducer behaves on smooth sequences:
\begin{lemma}
    \label{lem:smooth-same-type}
    Define {\em the type} of an extended state to be the information 
    whether the state belongs to $\overleftarrow Q$,
    belongs to $\overrightarrow Q$, is equal to $\checkmark$, is equal to $\xmark$, or is equal to $\lightning$.
    If a sequence $m_1 \ldots m_k$ is smooth, 
    then  $\bar q m_1$ has the same type as $\bar q m_1 \cdots m_k$,
    for every $\bar q \in \bar Q$
\end{lemma}
\begin{proof}
    If $\bar q m_1 \not \in \overrightarrow Q$, then the lemma follows from the definition of the
    monoid action -- $(\bar q m_1) m_2 \cdots m_k = \bar q m_1$.
    If $\bar q m_1 \in \overrightarrow Q$,
    apply Green's Eggbox Lemma (Lemma \ref{lem:green-eggbox}) to obtain such $m' \in M$ that
    \begin{align*}
        m_1 \cdots m_k m' = m_1
    \end{align*}
    The existence of such $m'$ proves that $\bar q m_1 \ldots m_k \in \overrightarrow Q$, because
    otherwise, from the definition of the right action:
    \begin{align*}
    \bar q m_1 = (\bar q m_1 \cdots m_k) m' = \bar q m_1 \ldots m_k \not \in \overrightarrow Q
    \end{align*}
\end{proof}
The following lemma shows that, when running on smooth sequences, the two-way single-use transducer,
behaves almost like a one-way transducer:
\begin{lemma}
    \label{lem:smooth-run-monotonic}
    Take any word over the input alphabet $w \in \Sigma ^ *$. If $w = a_1 a_2 \ldots a_k$, such that
    $h(a_1), h(a_2), \ldots ,h(a_k)$ is a smooth sequence then for every $\bar q \in \bar Q$, the
    run $\bar q w$ is {\em monotonic} -- once it crosses the border between $a_i$ and $a_{i+1}$,
    it will never cross the border between $a_{i-1}$ and $a_{i}$ (it might however, cross each of the
    borders multiple times).
    \mypic{31}
\end{lemma}
\begin{proof}
    Take any $\bar q$. We assume that $\bar q \in \overrightarrow Q$,
    or otherwise the run would never enter $w$, making the lemma vacuously true.
    Suppose that the run starting in $\bar q$ is not monotonic.
    Then there exists an $i \in \set{1, \ldots, k}$ such that, the run crosses the border between $a_i$ and $a_{i+1}$ and then crosses the border between $a_{i-1}$ and $a_i$.
    \mypic{32}
    In the picture above this happens for $i=3$. Define $\bar p$ to be the state in which the transducer
    crosses the border between $a_{i-1}$ and $a_{i}$  border for the last time before crossing the border between $a_i$ and $a_{i+1}$ for the first time.
    It follows that $\bar p a_i \in  \overrightarrow Q$,
    but $\bar p a_i \cdots a_k \in \overleftarrow Q$,
    because the run on $w$ will cross the border between $a_{i-1}$ and  $a_i$ again,
    leaving the word $a_i \ldots a_k$ from the left. Since $a_i \cdots a_k$ is smooth,
    this contradicts Lemma \ref{lem:smooth-same-type}.
\end{proof}
Note that Lemmas \ref{lem:smooth-same-type} and \ref{lem:smooth-run-monotonic} are also true
in their versions for the 
left action of $M$ on $\bar Q$.

Thanks to Lemma~\ref{lem:extend-to-star},  we assume that $h$ satisfies condition (*) in the Split Lemma. 
We now state the main lemma of this section:
\begin{lemma}
    \label{lem:two-way-intermediate-steps}
    For each $k \in \set{1, 2, \ldots}$, the following is a single-use rational function:
    \begin{itemize}
        \item {\bf Input:} A word over the alphabet $\bar Q^{\leq k} + \Sigma$:
        \begin{align*}
        \bar Q_0 a_1 a_2 a_3 \ldots a_n \bar Q_n
        \end{align*}
        such that first and last letters -- denoted by $\bar Q_0$ and $\bar Q_n$ -- belong to $\bar Q ^{\leq k}$ and all the other letters belong to $\Sigma$.
        Extended states in $\bar Q_0$ may be masked (as in Definition \ref{def:mased-states}), but
        they have to be right-compatible with the product $a_1 \cdots a_n$. Similarly, extended
        states in $\bar Q_n$ may be masked, but have to be left-compatible with $a_1 \cdots a_n$.
        \item {\bf Output:} A word over the alphabet $(\bar Q \times (\{1, \ldots k\} \times \{-1, 0, 1\} + \bot))^{\leq k} + \Sigma$:
        \begin{align*}
            \bar Q_0 a_1 \bar Q_1 a_2 \bar Q_2 a_3 \bar Q_3 \ldots \bar Q_{n-1} a_n \bar Q_n
        \end{align*}
        such that every $\bar Q_i$ contains all the extended states in which a run of the automaton
        that starts in one of the extended states from $\bar Q_0$ or from $\bar Q_n$ and
        that runs on the sequence $a_1, \ldots, a_n$,
        will cross the border between $a_i$ and $a_{i+1}$. 
        If there is more than $k$ of such extended states,
        the behaviour of the function is undefined.
        Moreover, each of
        those intermediate extended states should be equipped with the information about
        its successor in the same format as in the alphabet for the run graphs:
        \begin{align*}
        (\textrm{column offset of the successor}, \textrm{row of the successor}) \textrm{ or $\bot$, if there is no successor}
        \end{align*}
        Extended states in $\bar Q_i$ may be masked (even more masked than the states in $\bar Q_0$ and $\bar Q_n$), but they have to be
        right-compatible with $a_{i+1}$ and left-compatible with $a_{i}$. 
    \end{itemize}
\end{lemma}
First, we note that this lemma implies that the translation of a word to its run graph is a
single-use rational function. To compute the run graph for an input string $w=a_1 \cdots a_n$, do the following:
\begin{enumerate}
    \item take $k$ equal to the bound on the width of the run graph from Lemma~\ref{lem:bounded-width};
    \item  compute the Shepherdson function for the input word as the
           homomorphic image of the entire word (using the same technique
           as in Section \ref{sec:moniod-to-automaton} of the appendix);
    \item based on the Shepherdson function for the entire word, calculate the tuples of states
          in which the automaton will cross the border between $\vdash$ and $a_1$
          (call the tuple $\bar Q_0$), and the border between $a_n$ and $\dashv$ (call the tuple $\bar Q_n$);
    \item write $\bar Q_0$ at the beginning and $\bar Q_n$ at the end of the word;
    \item apply Lemma \ref{lem:two-way-intermediate-steps};
    \item based on each $Q_i$, $a_i$, $a_{i+1}$ calculate the local part of the run graph --
          this is possible because Lemma \ref{lem:two-way-intermediate-steps} guarantees that the
          states in $\bar Q_i$ are compatible with both $h(a_i)$ and $h(a_{i+1})$;
\end{enumerate}
The rest of this section is dedicated to proving Lemma~\ref{lem:two-way-intermediate-steps}:
Induction on the height of the smooth split produced by Split Lemma almost immediately 
reduces the general case to the case
where the sequence represented by $a_1 \ldots a_n$ is smooth. Thanks to the monotonicity of the runs
(Lemma \ref{lem:smooth-run-monotonic}), we can compute the intermediate states
of smooth runs in two passes -- one from left to right which computes all the successors of $Q_0$
and one from right to left which we computes all the successors of $Q_n$.
Since the passes are symmetrical, we only describe the single-use one-way transducer
responsible for the left-to-right pass (call it $\mathcal{T}$)

The left-to-right transducer $\mathcal{T}$ has two buffers
called \texttt{current} and \texttt{next}. Each of them can keep up to $k$
extended states. This requires a limited classical memory (number of states) and a limited
number of registers. To satisfy the single-use condition the buffers need to keep
each extended state in two copies -- one copy will be used to
simulate the run of the original two-way automaton, and the other copy
will be used for the output. Moreover, in order to keep track
of the successors, the buffers should be ordered --
every extended state pushed to the buffer gets a fixed position.
Buffer \texttt{current} also keeps track of which
extended states that it carries have already been processed
in the current step.
For the ones that have been processed, the buffer also keeps information about
their successors in the (offset, row) style. The transducer $\mathcal{T}$ works
in the following way (in the end of this section we discuss why this
construction is single use):

\begin{enumerate}
    \item Add each extended state from $\bar Q_0$ to \texttt{current} (preserving the order)
          and go to position with $a_1$.
    \item For each extended state $\bar q$ in \texttt{current}, check
          if $\bar q a_1 \in \overrightarrow Q$. If so, add $\bar q a_1$ to
          the \texttt{next} and note that this is the successor of $\bar q$,
          otherwise note that $\bar q$ has no successor.
    \item Output all the states from \texttt{current} together with
          information about their successors.
    \item Output $a_1$.
    \item Move all the extended states from \texttt{next} to \texttt{current} (this leaves \texttt{next} empty).
    \item For each $a_i$ for $i \in \set{2, 3, \ldots}$:
    \begin{enumerate}
        \item For each $\bar q \in \mathtt{current}$:
        \begin{enumerate}
            \item If $\bar q \in \overrightarrow Q$, compute $\bar q' = \bar q a_i$. If $\bar q' \in \overrightarrow Q$
                  add $\bar q'$ to \texttt{next} and keep the information about its successor (offset=$1$). Otherwise, if $\bar q' \in
                  \overleftarrow Q$ add $q'$ to \texttt{current} and keep the information about its successor
                  (offset=$0$).
            \item If $\bar q \in \overleftarrow Q$, compute
                 $\bar q' = a_{i-1} \bar q$. \footnote{Here, we assume that when in position $i$ the transducer has access to both $a_i$ and $a_{i-1}$.
                 This is not a problem because we can compose any transducer with the delay function
                 from Lemma \ref{lem:delay}} From
                  Lemma {\ref{lem:smooth-run-monotonic}} we know that $\bar q' \in \overrightarrow Q$.
                  Add $\bar q'$ to \texttt{current} and keep the information about its successor (offset=$0$).
            \item Otherwise, note that $\bar q$ has no successor.
        \end{enumerate}
        \item Output all the states from \texttt{current} together with information about their successors.
        \item Output $a_i$.
        \item Move all the extended states from \texttt{next} to \texttt{current}.
    \end{enumerate}
\end{enumerate}
In order for this construction to be single use, the transducer needs to be careful when computing values $\bar q'$,
so that it only uses those register values of $\bar q_j$ that are necessary to compute $\bar q_j a_i$
in the sense of the following lemma.
\begin{lemma}
    For every two words over the input alphabet of the two-way transducer $w_1, w_2 \in \Sigma^*$,
    such that $h(w_1) = h(w_2)$ and for every extended state $\bar q \in \bar Q$, the run $\bar q w_1$,
    uses exactly the same set of registers as the run $\bar q w_2$.  
\end{lemma}
\begin{proof}
    Suppose that for some $\bar q \in \bar Q$, the run $\bar q w_1$ uses some register
    that the run $\bar q w_2$ doesn't use. By setting the value of this register to $\bot$, we obtain
    $\bar q' \in \bar q \downarrow$, such that $\bar q'$ and $h(w_2)$ are compatible, but $\bar q'$ and $h(w_1)$
    are not. This contradicts the assumption that $h(w_1) = h(w_2)$. 
\end{proof}
$\mathcal{T}$ is able to compute $\bar q = \bar q m$
so that it only uses those register values from $\bar q$ that are used by $m$.
\footnote{One way to do it, is to use Lemma \ref{lem:straight-choice} to choose one of the shortest $w$, such
          that $f(w) = m$ and simulate the two-way automaton on $w$ starting in $\bar q$.
          In order to make the uniformisation equivariant, we require that all atoms in $w$
          that do not appear in $m$ are replaced with placeholders. This simulation can be done in bounded memory,
          because the length of such $w$ depends only on the orbit of $m$.} Thanks to that, all the
registers of $\bar q$ that were not used by $m$, can be moved to $\bar q'$.
Those registers are masked (replaced with $\bot$) in the version of $\bar q$
that is outputted by $\mathcal T$. All the registers that have been used by $m$
are stored in a second copy -- this copy is used to output their values. The same reasoning
can be used when computing the value $m \bar q$.

\subsection{Regular list functions to two-way transducers}
\label{ap:regular-list-functions-to-two-way}
In this section we show the implication \ref{two-way:regular-list-function} $\Rightarrow$ \ref{two-way:two-way} in Theorem~\ref{thm:two-way-models}, i.e.~the inclusion
\begin{align*}
    \twoinclusion
{regular list functions with atoms} {single-use two-way transducers.}
\end{align*}

The proof is a straightforward induction on the derivation of a regular list function. An element of a datatype, e.g.~
    \begin{align*}
    ([1,2,3], (0,[\bot]))  \in \atoms^* \times (\atoms \times \set \bot^*)
    \end{align*}
    can be viewed as a string over a polynomial  orbit-finite  alphabet that  consists of parentheses, commas and elements of the underlying polynomial orbit-finite sets. We refer to this description as the  \emph{string representation} of a datatype.  By a straightforward induction over the derivation, one  shows that for every  regular list function with atoms, there is a two-way single-use transducer which transforms string representations of inputs into string representations of outputs.
    The only non-trivial part of the proof is the combinator for function composition. Functions computed by two-way single-use transducers are closed under composition thanks to the equivalence 
    \begin{align*}
        \twoequality{two-way single-use transducers} {(two-way primes)*}
    \end{align*}
    that was already proved in Sections~\ref{ap:composition-to-two-way} and~\ref{ap:two-way-to-composition}.
\subsection{Streaming string transducers to two-way transducers}
\label{ap:sst-to-two-way}
In this section we show the implication \ref{two-way:sst} $\Rightarrow$ \ref{two-way:two-way} in Theorem~\ref{thm:two-way-models}, i.e.~the inclusion
\begin{align*}
    \twoinclusion
{\sst with atoms} {single-use two-way transducers.}
\end{align*}
This inclusion  is proved in the same way as in the case without atoms, see~\cite[Lemma 14.4]{bojanczyk_automata_2018}.
Fix an \sst with atoms. For an input string, define its \emph{register forest} as follows: Its nodes are the string register actions of the accepting run; they are either of the form $A := a$ for some letter $a$ in the output alphabet, or of the form $A := BC$.  An action of the first kind is a leaf of the register tree, and an action of the second kind has two children -- the most recently performed actions on the registers $B$ and $C$ respectively. Here is a picture of the register tree for the  transducer from Example~\ref{ex:iterated-reverse-sst}:
\mypic{3}
A register forest can be represented as a string over a polynomial orbit-finite alphabet (roughly speaking, the columns in the above picture),  and the transformation 
\begin{align}\label{eq:compute-register-tree}
    \text{input string} \ \mapsto \  \text{register tree}
\end{align}
can be computed by a single-use two-way (in fact, one-way)  transducer  which performs the same actions as the \sst, except that the operations on the string registers are ``printed out'' to the output. Two-way orbit-finite transducers  are closed under composition -- this is a corollary of already proved equality
\begin{align*}
    \twoequality{two-way single-use transducers} {(two-way primes)*},
\end{align*}
So it remains to show that the transformation
\begin{align*}
    \text{register tree} \ \mapsto \  \text{output string}
\end{align*}
is computed by a two-way orbit-finite transducer. This is done in the same way as without atoms: by doing a depth-first search traversal of the register forest, see~\cite[Lemma 14.4]{bojanczyk_automata_2018}.

\section{Equivalence}
\label{ap:equivalence}
In this part of the appendix, we prove Theorem~\ref{thm:equivalence}, which says that the equivalence problem is decidable for  \sst with atoms.
Our algorithm uses a reduction to the case without atoms. However, we 
do not reduce to the equivalence problem for \sst without atoms, whose equivalence problem is in {\sc pspace}~\cite[Theorem 12]{alurCerny11}. Instead, our reduction needs a stronger result, namely  equivalence for \emph{multiple-use} \sst without atoms~\cite[Theorem 3]{filiotCopyfulStreamingString2017}. The latter problem is decidable,   but its complexity is the same as for the HTD0L problem, which is unknown. 

\newcommand{\morsealph}{\set{\diamond,\circ}}
\newcommand{\morseset}{\diamond^*\circ}
\newcommand{\morse}{\diamond}
\smallparagraph{Deatomisation}
Let $\morsealph$ be a fresh alphabet, which does not intersect any other alphabet considered. We will represent atoms using strings of the form $\diamond^* \circ$ --
thanks to the  $\circ$-separator, every concatenation of such representations can be parsed unambiguously.   Define a \emph{deatomisation} to be any function $\alpha : \atoms \to \morse^* \circ$, not necessarily injective. Typically such a function will not be not finitely-supported, since otherwise it would have finite range.  Using a  deatomisation $\alpha$, we can represent letters from a polynomial orbit-finite alphabet as strings over a finite alphabet, as illustrated in the following example
\begin{align*}
    \underbrace{(3,1,1,3)}_{\in \atoms^4}\qquad \mapsto \qquad \morse \morse \morse \circ \, \morse  \circ \,   \morse \circ \  \morse \morse \morse \circ \, \qquad \text{assuming $1 \stackrel \alpha \mapsto \morse \, \circ$ and $3 \stackrel \alpha \mapsto \morse \morse \morse \, \circ$.}
    \end{align*}
More formally, for every orbit-finite alphabet $\Sigma$ we define a finite alphabet $\red \Sigma$, and lift the deatomisation $\alpha$ to a function $\alpha_\Sigma : \Sigma \to \red \Sigma^*$. This representation is defined in the natural way by induction on the structure of $\Sigma$, with the case of products explained in the example above.
For a string $w \in \Sigma^*$, we define its {deatomisation} to be the concatenation of deatomisations of the corresponding letters.

The following lemma is where multiple-use \sst appear.

\begin{lemma}\label{lem:deatomise-image}
    Let $g : \Sigma^* \to \Gamma^*$ be a single-use one-way transducer. The language of all
    possible deatomisations of the
    range of $g$:
    \begin{align*}
    \set{ \alpha_\Gamma(g(w)) : \textrm{for all }w \in \Sigma^* \text{and for all deatomisations $\alpha$}} \subseteq \red \Gamma^*
    \end{align*}
    is equal to the range of some multiple-use \sst $\red g : \red \Sigma^* \to \red \Gamma^*$ without atoms.
\end{lemma}
It is important that we also include deatomisations $\alpha$'s that are not injective -- this means that different atoms might get represented the same way after applying $\alpha$.
\begin{proof}
    We prove a stronger version of the Lemma, which works for any
    one-way transducer with atoms $g$ -- not necessarily single-use. We start with the following claim, which says that it suffices to show a nondeterministic multiple-use SST $\red g$ with the appropriate range:
    \begin{claim}
    For every nondeterministic multiple-use
    SST $\red g : \red \Sigma^* \to \red \Gamma^*$, there is a deterministic multiple-use SST $\red g' : \red \Sigma^* \to \red \Gamma^*$, such that the ranges of the two transducers are equal.
    \end{claim}
    \begin{proof}
        We assume that $\Sigma$ contains atoms -- if not then $g$
        already is atomless, and Lemma \ref{lem:deatomise-image} is
        immediate. This means that $\red \Sigma$ contains at least
        two letters -- $\circ$ and $\diamond$. Thanks to that $\red g'$ can simulate all the non-deterministic choices of $g$ by looking at the input letters -- this is going to make it ``ignore''
        parts of the input, but since the domain is
        the entire set $\red \Sigma^*$,
        it will in the end simulate every run of $\red g$ on every
        input word. 
    \end{proof}
    
    We are left with constructing appropriate nondeterministic multiple-use SST without atoms $\red g: \red \Sigma ^* \to \red \Gamma ^*$. To simplify the proof, we restrict ourselves to the special case where the input alphabet is
    atoms -- other cases are analogous. The idea is that $\red g$ simulates a run of $g$. The atom registers of $g$ are simulated by string registers of $\red g$, which store deatomisations of the original atoms. The transducer $\red g$ also remembers in its state an equality type
    of the contents of its registers.  
    Suppose that the original transducer $g$ is about to read some input atom $a$. Instead of looking at the deatomisation given in the input, the simulating transducer $\red g$ guesses which of the current registers will be equal to the newly read atom. If there is some current register that matches it, a deatomisation of the newly read atom is simply copied from the matching register. (Since there is no bound on the number of such copies used in the run, we need to have a multiple-use \sst). If the newly read atom is guessed to be fresh with respect to the current registers, then  $\red g$ uses nondeterminism to generate a fresh string that will represent the new atom (it could be the case that the nondeterminically guessed string happens to be equal to one of the registers -- we  call this a ``spurious equality'' and deal with it below). When comparing two registers, the automaton uses the equality-type information stored in its state. The other actions and questions are simulated in the natural way. 

The only problematic part of the construction described above is the ``spurious equality''. It could be the case that the configuration of $g$ has register valuation $(1,2)\in \atoms^2$ , the simulating transducer stores in its state that all atom registers are  distinct, but nevertheless the register contents of $\red g$ are not distinct, e.g.~$(\diamond \diamond\circ, \diamond \diamond\circ)$, because when guessing a new string the transducer accidentally guessed the same string twice. This is not a problem, because we can view $(\diamond \diamond \circ, \diamond \diamond\circ)$ as a non-injective deatomisation of $(1,2)$. 
\end{proof}

\smallparagraph{Equivalence}
We now proceed to describe the algorithm for checking equivalence of \sst with atoms. Consider an input to the equivalence problem, which consists of two functions 
\begin{align*}
    f_1,f_2 : \Sigma^* \to \Gamma^*
\end{align*}
given by  \sst with atoms. We want to check if the functions are equal. Recall the register forests  discussed in Section~\ref{ap:sst-to-two-way}. The register forest for each $f_i$ can be computed by a single-use one-way transducer, and these two register  forests can be represented as a single word over the product alphabet, as
shown in the following picture:
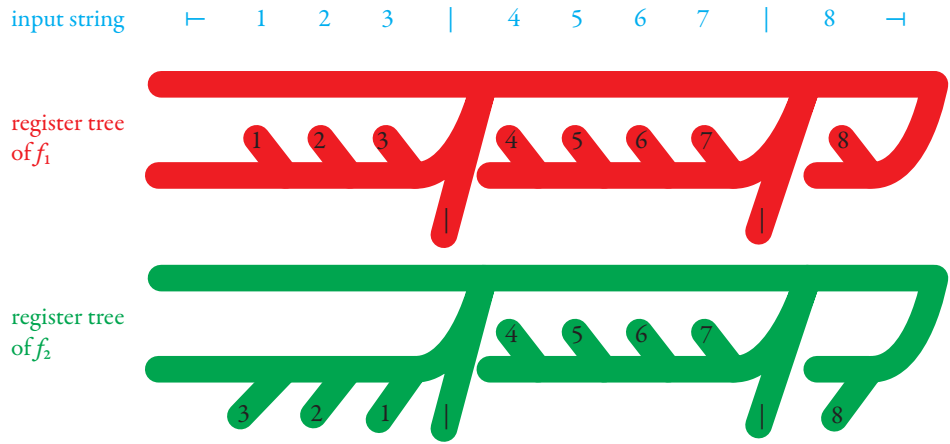
\mypic{27}
Let $g : \Sigma^* \to \Delta^*$ be a single-use one-way transducer which computes the combined register forests of both transducers. For $i \in \set{1,2}$, let  $g_i : \Delta^* \to \Gamma^*$ be an \sst with atoms which transforms the combined register forest produced by $g$ into the output of $f_i$. These transformations can be described in the following diagram, where the upper and lower (but not necessarily middle) faces commute:
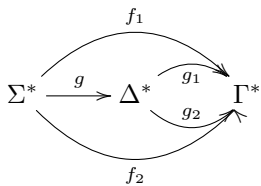
\begin{align*}
    \xymatrix{
        \\
        \Sigma^*
        \ar@/^2.0pc/[rr]^{f_1}
        \ar@/_2.0pc/[rr]_{f_2} 
        \ar[r]^g 
        & \Delta^*
        \ar@/^1.0pc/[r]_{g_1}
        \ar@/_1.0pc/[r]^{g_2}
         & \Gamma^* 
    }
    \end{align*} 
 The equivalence $f_1=f_2$ is the same as the equivalence $g_1=g_2$ restricted the range of $g$.  Our algorithm will check the latter condition.

The following lemma shows that $g_1$ and $g_2$ can be represented, along deatomisation, using \sst without atoms. 
The  lemma crucially depends on the fact that
    computing the output of a register forest does not require any equality tests on the input -- it only moves letters around. A similar construction does not work for an arbitrary \sst with atoms (or we would have applied it directly to $f_1$ and $f_2$). 

\begin{lemma}\label{lem:equality-oblivious}
    Let $i \in \set{1,2}$. There is an \sst without atoms  $\red{g_i} : \red \Delta^* \to \red \Gamma^*$ such that the following diagram commutes for every deatomisation $\alpha$:
    \begin{align*}
        \xymatrix{
            \Delta^* \ar[r]^{g_i} \ar[d]_{\alpha} &
            \Gamma^* \ar[d]^{\alpha} \\
            \red \Delta^* \ar[r]_{\red{g_i}} &
            \red \Gamma^*
        }
    \end{align*}
\end{lemma}
\begin{proof} A straightforward simulation, that keeps the deatomisations of atoms in the string registers. (This is why it is important that there are no equality tests, since otherwise one would need to check equality for string registers.)
\end{proof}

Apply  Lemma~\ref{lem:deatomise-image} to $g$, yielding  a multiple-use \sst without atoms -- $\red g$. Its range is the same as the deatomisations of the range of $g$.
\begin{lemma}
    $f_1=f_2$ if and only if $\red {g_1} \circ \red g = \red{g_2} \circ \red g$.
\end{lemma}
\begin{proof}
    The left-to-right implication is immediate. For the converse implication, we prove 
    \begin{align*}
        f_1 \neq f_2 \qquad \text{implies} \qquad \red {g_1} \circ \red g \neq \red{g_2} \circ \red g.
    \end{align*}
    Assume $f_1 \neq f_2$, and choose some input $w \in \Sigma^*$ witnessing the inequality. Choose some \emph{injective} deatomisation $\alpha$.
    This gives us
    \begin{align}\label{eq:unequal-deatomisation}
    \alpha_\Gamma(f_1(w)) \neq \alpha_\Gamma(f_2(w)).
    \end{align}
    By Lemma~\ref{lem:deatomise-image}, there is some input $u$ to the function $\red g$ such that 
    \begin{align}\label{eq:in-the-image}
    \red g(u) = \alpha_\Delta (g(w)).
    \end{align}
    For $i \in \set{1,2}$ we have the following equality:
    \begin{align*}
     \red{g_i}(\red g (u)) \stackrel{\text{\eqref{eq:in-the-image}}} =  \red{g_i}(\alpha_\Delta (g(w))) \stackrel {\text{Lemma~\ref{lem:equality-oblivious}}} =  \alpha_\Gamma(g_i(g(w)))  = \alpha_\Gamma(f_i(w))
    \end{align*}
    Combining this with~\eqref{eq:unequal-deatomisation}, we see that $u$ is an argument for which the functions  $\red {g_1} \circ \red g$ and $\red {g_2} \circ \red g$  give different outputs. 
\end{proof}

Thanks to this lemma, the equivalence problem $f_1=f_2$  for  infinite alphabets reduces to the equivalence problem 
\begin{align}\label{eq:reduced-equivalence}
    \red {g_1} \circ \red g = \red{g_2} \circ \red g,
\end{align}
which uses finite alphabets. Each side  of the above equivalence is a function in the class 
\begin{align*}
    \text{(\sst without atoms)} \circ \text{(multiple-use \sst without atoms)}.
    \end{align*}
At the end of  Section~\ref{ap:composition-to-sst}, we proved that this class is (effectively) contained  in multiple-use \sst without atoms (the proof relied on the Krohn-Rhodes decomposition for \sst). Therefore~\eqref{eq:reduced-equivalence} is an instance of the equivalence problem for multiple-use \sst without atoms, which is decidable.

\end{document}